\documentclass[a4paper,UKenglish]{lipics-v2018}
\usepackage{microtype}%if unwanted, comment out or use option "draft"

\nolinenumbers % to switch line numbering off (globally)

\title{Multivariate Submodular Optimization}

\author{Richard Santiago}{School of Computer Science, McGill University, Montreal, Canada}{richard.santiagotorres@mail.mcgill.ca}{}{}
\author{F. Bruce Shepherd}{Department of Computer Science, University of British Columbia, Vancouver, Canada}{fbrucesh@cs.ubc.ca}{}{}
\authorrunning{R. Santiago and F. B. Shepherd}%mandatory. First: Use abbreviated first/middle names. Second (only in severe cases): Use first author plus 'et al.'

\Copyright{Richard Santiago and F. Bruce Shepherd}%mandatory, please use full first names. LIPIcs license is "CC-BY";  http://creativecommons.org/licenses/by/3.0/

\subjclass{}% mandatory: Please choose ACM 2012 classifications from https://www.acm.org/publications/class-2012 or https://dl.acm.org/ccs/ccs_flat.cfm . E.g., cite as "General and reference $\rightarrow$ General literature" or \ccsdesc[100]{General and reference~General literature}. 
%mandatory

\makeatletter
\newtheorem*{rep@theorem}{\rep@title}
\newcommand{\newreptheorem}[2]{%
	\newenvironment{rep#1}[1]{%
		\def\rep@title{#2 \ref{##1}}%
		\begin{rep@theorem}}%
		{\end{rep@theorem}}}
\makeatother

\theoremstyle{plain}
\newtheorem{claim}[theorem]{Claim}
\newtheorem{proposition}[theorem]{Proposition}
\newreptheorem{theorem}{Theorem}

\newcommand{\R}{\mathbb{R}}

\newcommand{\E}{\mathbb{E}}
\newcommand{\Z}{\mathbb{Z}}

\newcommand{\F}{\mathcal{F}}

\renewcommand{\P}{\mathcal{P}}

\newcommand{\B}{\mathcal{B}}
\renewcommand{\H}{\mathcal{H}}
\renewcommand{\L}{\mathcal{L}}

\newcommand{\T}{\mathcal{T}}

\newcommand{\MA}{multi-agent}

\renewcommand{\S}{(S_1,\ldots,S_k)}
\renewcommand{\T}{(T_1,\ldots,T_k)}

\DeclareMathOperator*{\argmin}{argmin}
\DeclareMathOperator*{\argmax}{argmax}

\newcommand{\notni}{\not\ni}

\EventEditors{John Q. Open and Joan R. Access}
\EventNoEds{2}
\EventLongTitle{42nd Conference on Very Important Topics (CVIT 2016)}
\EventShortTitle{CVIT 2016}
\EventAcronym{CVIT}
\EventYear{2016}
\EventDate{December 24--27, 2016}
\EventLocation{Little Whinging, United Kingdom}
\EventLogo{}
\SeriesVolume{42}
\ArticleNo{23}
\hideLIPIcs  %uncomment to remove references to LIPIcs series (logo, DOI, ...), e.g. when preparing a pre-final version to be uploaded to arXiv or another public repository
\begin{document}
	
\maketitle

\begin{abstract}
	Submodular functions have found a wealth of new applications in data science and machine learning models in recent years.
	This has been coupled with many algorithmic advances in the area of submodular optimization: (SO) $\min/\max~f(S): S \in \mathcal{F}$, where $\mathcal{F}$ is a given family of feasible sets over a ground set $V$ and $f:2^V \rightarrow \mathbb{R}$ is submodular.
	In this work we focus on a more general class of \emph{multivariate submodular optimization} (MVSO) problems:  $\min/\max~f (S_1,S_2,\ldots,S_k):  S_1 \uplus S_2 \uplus \cdots \uplus S_k \in \mathcal{F}$. Here we use $\uplus$ to denote disjoint union and hence this model is attractive where resources are being allocated across  $k$ agents, who share a ``joint'' multivariate nonnegative objective $f(S_1,S_2,\ldots,S_k)$ that captures some type of submodularity (i.e. diminishing returns) property.
	We provide some explicit examples and potential applications for this new framework.
	
	For maximization, we show that practical algorithms such as accelerated greedy variants and distributed algorithms achieve good approximation guarantees for very general families (such as matroids and $p$-systems). For arbitrary families, we show that monotone (resp. nonmonotone) MVSO admits an $\alpha (1-1/e)$ (resp. $\alpha \cdot 0.385$) approximation whenever monotone (resp. nonmonotone) SO admits an $\alpha$-approximation over the multilinear formulation. This substantially expands the family of tractable models for submodular maximization.
	For minimization, we show that if SO admits a $\beta$-approximation over \emph{modular} functions, then MVSO admits a $\frac{\beta \cdot n}{1+(n-1)(1-c)}$-approximation where $c\in [0,1]$ denotes the curvature of $f$. We show that this approximation is essentially tight even for $\mathcal{F}=\{V\}$. Finally, we give a bound in terms of $k$ and prove that MVSO has an $\alpha k$-approximation whenever SO admits an $\alpha$-approximation over the convex formulation.
	%On the minimization side we give essentially optimal approximations in terms of the curvature of $f$.
\end{abstract}

\newpage

%%%%%%%%%%%%%%%%%%%%%%%%%%%%%%%%%%%%%%%%%%%%%%%%%%%%%%%%%%%%%%%%%%%%%%%%%%%%%%%%%%%%%%%%%%%%%%%%%%%%%%%%%%%%%%%%%%%%%%%%%%%%%%%%%%%%%%%%%%%%%%%%%%%%%%%%%%%%%%
%%%%%%%%%%%%%%%%%%%%%%%%%%%%%%%%%%%%%%%%%%%%%%%%%%%%%%%%%%%%%%%%%%%%%%%%%%%%%%%%%%%%%%%%%%%%%%%%%%%%%%%%%%%%%%%%%%%%%%%%%%%%%%%%%%%%%%%%%%%%%%%%%%%%%%%%%%%%%%
\section{Introduction}
\label{sec:intro}
%%%%%%%%%%%%%%%%%%%%%%%%%%%%%%%%%%%%%%%%%%%%%%%%%%%%%%%%%%%%%%%%%%%%%%%%%%%%%%%%%%%%%%%%%%%%%%%%%%%%%%%%%%%%%%%%%%%%%%%%%%%%%%%%%%%%%%%%%%%%%%%%%%%%%%%%%%%%%%
%%%%%%%%%%%%%%%%%%%%%%%%%%%%%%%%%%%%%%%%%%%%%%%%%%%%%%%%%%%%%%%%%%%%%%%%%%%%%%%%%%%%%%%%%%%%%%%%%%%%%%%%%%%%%%%%%%%%%%%%%%%%%%%%%%%%%%%%%%%%%%%%%%%%%%%%%%%%%%

%\iffalse
Submodularity is a property of set functions with deep theoretical
consequences and a wide range of applications. Optimizing submodular
functions is a central subject in operations research and combinatorial
optimization \cite{lovasz1983submodular}. It appears in many important
optimization frameworks including cuts in graphs, set covering problems, plant location problems, certain
satisfiability problems, combinatorial auctions, and maximum entropy
sampling. In machine learning it has recently been identified and
utilized in domains such as viral marketing \cite{kempe2003maximizing},
information gathering \cite{krause2007near}, image segmentation \cite{boykov2001interactive,kohli2009p3,jegelka2011submodularity},
document summarization \cite{lin2011class}, and speeding up satisfiability
solvers \cite{streeter2009online}.

A set function $f:2^V \to \R$ is \emph{submodular} if $f(S) + f(T) \geq f(S \cup T) + f(S \cap T)$ for any $S,T \subseteq V$.
We say that $f$ is \emph{monotone} if $f(S) \leq f(T)$ for $S \subseteq T$.
Throughout, all submodular functions are nonnegative, and we usually assume $f(\emptyset)=0$. Our functions are  given by a \emph{value oracle}, where for a given set $S$ an algorithm can query the oracle to find its value $f(S)$.
%\fi

We consider the following broad class of submodular optimization (SO) problems:
\begin{equation}
\label{eqn:SA}
\mbox{SO($\F$) ~~~~Min~ / ~Max ~$f(S):S\in\F$}
\end{equation}
where $f$ is a nonnegative submodular set function on a finite ground set $V$,
and $\F \subseteq 2^V$ is a family of feasible sets.  
%There has been a large recent stream  of activity around these problems for a variety of set families $\F$.
These problems have been well studied for a variety of set families $\F$.
We explore  the connections between these (single-agent) problems and their more general multivariate
incarnations.
In the {\em multivariate (MV)} version, we have $k$ agents and a ``joint'' multivariate
nonnegative objective $f(S_1,S_2,\ldots,S_k)$ that captures some type of submodularity (i.e.
diminishing returns) property
%(we provide the formal definition below).
(see Section \ref{sec:mv}). 
As before, we are looking for  sets $S\in\F$, however,
we now have a 2-phase task: the elements of $S$ must also be partitioned amongst the agents.
Hence we have set variables $S_{i}$ and seek to optimize $f(S_1,S_2,\ldots,S_k)$.
This leads to the multivariate submodular optimization (MVSO) versions:
\begin{equation}
\label{eqn:MV}
\mbox{MVSO($\F$) ~~~~Min~ / ~Max ~$f(S_1,S_2,\ldots,S_k):S_{1}\uplus S_{2}\uplus \cdots\uplus S_{k}\in\F$.}
\end{equation}

%Our main objective is to explain  the extent to which approximability
%for multivariate problems is intrinsically connected to their single-agent versions.
Our main objective is to study the approximability of the multivariate problems in terms of their single-agent versions.
We refer to the {\em multivariate (MV) gap} as the approximation factor loss incurred
by moving to the multivariate setting. 
To the best of our knowledge, neither the MVSO($\F$) framework for general families $\F$ nor the notion of MV gap have been considered before in the literature.

An important special case of MVSO occurs when the function $f \S$ can be
\emph{separated} as $f \S = \sum_{i\in [k]} f_i (S_i)$ where the $f_i$ are all submodular;
in this case we say that $f$ is \emph{separable}. This leads to the
class of multi-agent submodular optimization (MASO) problems
\begin{equation}
\label{eqn:MA}
\mbox{MASO($\F$) ~~~~Min~ / ~Max ~$\sum_{i=1}^{k} f_{i}(S_{i}):S_{1}\uplus S_{2}\uplus \cdots\uplus S_{k}\in\F$,}
\end{equation}
which have been widely studied (see related work section).

\iffalse
The special case when $\F=\{V\}$ has been of great interest.
It has been examined for both
minimization  (the minimum submodular cost allocation problem \cite{hayrapetyan2005network,svitkina2010facility,ene2014hardness,chekuri2011submodular}) and
maximization (submodular welfare problem \cite{fisher1978analysis,lehmann2001combinatorial,vondrak2008optimal}).
It has also been studied for more general families $\F$ \cite{goel2009approximability,santiago2018ma}.
\fi

\subsection{Multivariate submodular optimization}
\label{sec:mv}

We consider functions of several variables which satisfy the following type of submodularity property. A multivariate function $f:2^{kV}\to\R$ is \emph{$k$-multi-submodular}
if for all pairs of tuples $(S_{1},S_2,...,S_{k}),(T_{1},T_2,...,T_{k})\in2^{kV}$
we have
\[
f(S_{1},...,S_{k})+f(T_{1},...,T_{k})\geq f(S_{1}\cup T_{1},S_{2}\cup T_{2},..,S_{k}\cup T_{k})+f(S_{1}\cap T_{1},S_2 \cap T_2,...,S_{k}\cap T_{k}).
\]
Moreover, we say that $f$ is \emph{normalized} if $f(\emptyset,\emptyset,\ldots,\emptyset)=0$, and
\emph{monotone} if $f \S \leq f \T$ for all tuples $\S$ and $\T$ satisfying $S_i \subseteq T_i$ for all $i \in [k]$.

In the special case of $k=1$ a $k$-multi-submodular function is just a submodular function. 
%Hence, the above can be thought of as a natural extension of submodularity to multivariate functions. 
In Appendix \ref{sec:properties-mv-functions} we discuss how $k$-multi-submodular
functions can also be naturally characterized (or defined) in terms of diminishing returns.
This notion of multivariate submodularity has been considered before (\cite{fisher1978analysis,singh2012bisubmodular}) and we discuss this in detail on Section \ref{sec:related-work}.

Two explicit examples of (non-separable) $k$-multi-submodular functions (see Appendix \ref{sec:examples-mv-functions} for proofs)
are the following.
\begin{example}
	Consider a multilinear function $h:\Z^k_+ \to \R$ given by $h(z)=\sum_{S \subseteq [k]} a_S \prod_{m \in S} z_m$. Let $f:2^{kV} \to \R$ be a multivariate set function defined as $f(S_1,\ldots,S_k)=h(|S_1|,\ldots,|S_k|)$. Then $f$ is $k$-multi-submodular if and only if $a_S \leq 0$ for all $S \subseteq [k]$.
\end{example}

\begin{example}
	Let $h:\Z^k_+ \to \R$ be a quadratic function given by $h(z)=z^T A z$. Let $f:2^{kV} \to \R$ be a multivariate set function defined as $f(S_1,\ldots,S_k)=h(|S_1|,\ldots,|S_k|)$. Then $f$ is $k$-multi-submodular if and only if $A=(a_{ij})$ satisfies $a_{ij} + a_{ji} \leq 0$ for all $i,j \in [k]$.
\end{example}

We believe the above examples are useful for modelling ``competition'' between agents in many domains.
In Section \ref{sec:applications} we discuss one application to  sensor placement problems.

\subsection{Our contributions}

Our first contribution is to show that the MV framework can model much more general problems than the separable multi-agent (i.e. MASO) framework. This is quantitatively captured in the following information theoretic result (see Section~\ref{sec:hardness-monot}) where  we  establish a large gap between the two problems:
%We give complexity-theoretic evidence of this in Section~\ref{sec:hardness-monot} where  we  establish a large gap between the two problems:
\[
{\bf (MV-Min)}
\hspace{0.2cm}
\begin{array}{cc}
\min & f(S_{1},\ldots,S_{k})\\
\mbox{s.t.} & S_1 \uplus \cdots \uplus S_k=V%\\
%& S_{i}\subseteq V_{i}\\
%& |S_{i}| \leq 1
\end{array}
\hspace*{0.6cm}
{\bf (MA-Min)}
\hspace{0.2cm}
\begin{array}{cc}
\min & \sum_{i=1}^k f_{i}(S_{i})\\
\mbox{s.t.} & S_1 \uplus \cdots \uplus S_k=V%\\
%& S_{i}\subseteq V_{i}\\
%& |S_{i}| \leq 1
\end{array}
\]

\iffalse
We note that (MA-Min) with nonnegative monotone objectives $f_i$ is the submodular facility location
problem studied in \cite{svitkina2010facility}, where a tight $O(\log (n))$-approximation is given.
We show (see Theorem ---) that (MV-Min) is $\Omega(n)$-inapproximable under the oracle model.
\fi

\begin{theorem}
	\label{thm:MV-monot-hardness}
	The MV-Min problem with a nonnegative monotone $k$-multi-submodular objective function cannot be approximated to a ratio $o(n/\log n)$ in the value oracle model with polynomial number of queries, whereas its separable version MA-Min has a tight $O(\log n)$-approximation polytime algorithm for nonnegative monotone submodular functions $f_i$.
\end{theorem}

The above result shows that the MV model may also potentially face roadblocks in terms of tractability. 
Fortunately, we can show that the multivariate problem remains very well-behaved in the maximization setting.
%In the maximization setting, however, we show that the multivariate problem remains very well-behaved.
Our main result  establishes that if the single-agent problem for a family $\F$ admits approximation via its multilinear relaxation (see Section \ref{sec:max-SA-MA-formulations}),
then we may extend this to its multivariate version with a constant factor loss.

\begin{theorem}
	\label{thm:max-MA-gap}
	If there is a (polytime) $\alpha(n)$-approximation for monotone SO($\F$) maximization
	via its multilinear relaxation, then there is a (polytime) $(1-1/e) \cdot \alpha(n)$-approximation
	for monotone MVSO($\F$) maximization. Furthermore, given a downwards closed family $\F$,
	if there is a (polytime) $\alpha(n)$-approximation
	for nonmonotone SO($\F$) maximization
	via its multilinear relaxation, then there is a (polytime) $0.385 \cdot \alpha(n)$-approximation
	for nonmonotone MVSO($\F$) maximization.
\end{theorem}

We note that the multilinear relaxation can be
efficiently evaluated for a large class of practical and useful submodular functions \cite{iyer2014monotone}, thus making these algorithms viable for many real-world machine learning problems.

We remark that the MV gap of $1-1/e$ for monotone objectives is tight, in the sense that there are families where this cannot be improved.
For instance, $\F=\{V\}$ has a trivial $1$-approximation for the single-agent problem, and a $1-1/e$ inapproximability factor for the separable multi-agent (i.e. MASO) version \cite{khot2005inapproximability,mirrokni2008tight}, and hence also for the more general MVSO problem.

An immediate application of Theorem \ref{thm:max-MA-gap} is that it provides the first constant (and in fact optimal) ($1-1/e$)-approximation for the monotone \emph{generalized submodular welfare} problem $\max f(S_1,S_2,\ldots,S_k): S_1 \uplus \cdots \uplus S_k=V$. This problem generalizes the well-studied submodular welfare problem \cite{lehmann2001combinatorial,vondrak2008optimal,korula2018online}, which captures several allocation problems and has important applications in combinatorial auctions, Internet advertising, and network routing.
The MV objectives can capture much more general interactions among the agents/bidders, where now a bidder's valuation does not only depend
on the set $S$ of items that she gets, but also on the items that her strategic partners and competitors get. 
For instance, in a bandwidth spectrum auction, this could capture a company's interest to maximize compatibility and prevent cross-border interference.

\iffalse
In Section \ref{sec:MASA} we also describe a generic reduction that transforms any CMVSO($\F$) problem
(i.e. any multivariate problem (\ref{mv})) into a SO($\F'$) problem.
We use the idea of viewing assignments of elements $v$ to agents $i$ in a {\em multi-agent bipartite graph}.
This simple idea (which is equivalent to making $k$ disjoint copies of the ground set $V$) has been already used
in several previous works.
In \cite{santiago2018ma} we described this reduction for the simpler case where the objective
has the form $f \S = \sum_{i \in [k]} f_i(S_i)$.
\iffalse
This was first proposed
by Lehmann et al in \cite{lehmann2001combinatorial} and also used in Vondrak's work  \cite{vondrak2008optimal}. In those cases
$\F=\{V\}$ and $\F_i = 2^V$.
It has been also used in \cite{santiago2018ma} for more general families $\F$ and $\F_i$, and multivariate separable objectives.
\fi
Here we discuss the impact of the reduction on the more general setting of (non-separable) multivariate objectives,
where there is a priori no reason that this should be well-behaved.
Our results show that for some families an MV gap of $1$ holds.
\fi

In Section \ref{sec:MASA} we describe a simple reduction that
shows that for some families
\footnote{A family of sets $\F$ is a \emph{$p$-system} if for all $S \in \F$ and $v \in V$ there exists a set $T \subseteq S$
	such that $|T| \leq p$ and $S \setminus T \cup \{v\} \in \F$. A \emph{matroid} is a $1$-system.
	Cardinality and partition constraints are examples of matroids.
	%A non-empty downwards closed family of sets $\F$ is called a \emph{matroid} if for every $S,T \in \F$ with $|S|<|T|$, there exists $v \in T \setminus S$ such that $S \cup v \in \F$. 
	We refer the reader to \cite{schrijver2003combinatorial,calinescu2007maximizing,calinescu2011maximizing} for a comprehensive discussion.} %about matroids, matroid intersections, and $p$-systems.}
%For a set $S \subseteq E$, a set $B\subseteq S$ is called a \emph{basis} of $S$ if
%$B$ is an inclusion-wise maximal independent subset of $S$. 
%We say that $(V,\F)$ is a \emph{$p$-system} if for each $U \subseteq V$, the cardinality
%of the largest basis of $U$ is at most $p$ times the cardinality of the smallest basis of $U$.}
an (optimal) MV gap of $1$ holds. We also discuss how for those families, practical algorithms 
(such as accelerated greedy variants and distributed algorithms) can be used and lead to good approximation guarantees.

\begin{theorem}%[MV Gap $1$ Families]
	\label{thm:max-invariance}
	%Let $\F$ be either a $p$-system, a
	Let $\F$ be a matroid, a $p$-matroid intersection, or a $p$-system. Then, if there is a (polytime) $\alpha$-approximation algorithm for monotone (resp. nonmonotone) SO($\F$) maximization,
	there is a (polytime) $\alpha$-approximation algorithm for monotone (resp. nonmonotone) MVSO($\F$) maximization.
\end{theorem}

On the minimization side our approximation results and MV gaps are larger.
This is somewhat expected due to the strong hardness results already existing for single-agent submodular minimization (see Section \ref{sec:related-work}). 
However, we give essentially tight approximations in terms of the objective's curvature. The notion of curvature has been widely used for univariate functions \cite{conforti1984submodular,vondrak2010curvature,iyer2013curvature,bai2018greed}, since it allows for better approximations and it is linear time computable.

Given a tuple $(S_1,\ldots,S_k) \in 2^{kV}$ and $(i,v) \in [k]\times V$, we denote by $(S_1,\ldots,S_k)+(i,v)$ the new tuple $(S_1,\ldots,S_{i-1},S_i + v, S_{i+1},\ldots,S_k)$.
Then, it is natural to think of the quantity
\begin{equation*}
%\label{eq:mv-diminishing-returns}
f_{\S}((i,v)) := f(\S + (i,v)) - f \S
\end{equation*}
as the marginal gain of assigning element $v$ to agent $i$ in the tuple $(S_1,\ldots,S_k)$.
%We use the notation $f_{\S}((i,v))$ to denote the marginal gain in (\ref{eq:mv-diminishing-returns}).
We also use $f((i,v))$ to denote the quantity $f(\emptyset,\ldots,\emptyset,v,\emptyset,\ldots,\emptyset)$ where $v$ appears in the $ith$ component.
Then given a normalized monotone $k$-multi-submodular function $f:2^{kV}\to \R$ we define its \emph{total curvature} $c$ and its \emph{curvature $c(S_1,\ldots,S_k)$ with respect to a tuple} $(S_1,\ldots,,S_k) \subseteq V^k$ as
\begin{equation*}
c = 1 - \min_{i\in [k], v\in V} \frac{f_{(V,V,\ldots,V)-(i,v)}((i,v))}{f((i,v))} \hspace*{0.3cm}, \hspace*{0.3cm} c \S = 1 - \min_{ i\in [k], v \in S_i} \frac{f_{\S - (i,v)}((i,v))}{f((i,v))}.
\end{equation*}

We prove the following curvature dependent result for $k$-multi-submodular objectives.
We note the gap is stronger in the sense that it is relative to the single-agent {\em modular} problem.
\footnote{A function $f:2^{V}\to\R$ is modular if
	$f(A)+f(B)=f(A\cup B)+f(A\cap B)$ for all $A,B\subseteq V$. Modular
	functions can always be expressed in the form $f(S)=\sum_{v\in S}w(v)$
	for some weight function $w:V\to\R$.}
\begin{theorem}
	\label{thm:curvature}
	Let $f$ be a monotone $k$-multi-submodular function, and let $\F$ be a family that admits a (polytime) $\beta$-approximation over modular functions. Denote by $(S^*_1,\ldots,S^*_k)$ an optimal solution to monotone MVSO($\F$) minimization, and by $c(S^*_1,\ldots,S^*_k)$ the curvature of $f$ with respect to $(S^*_1,\ldots,S^*_k)$. Then there is a (polytime) $\frac{\beta \sum_{i \in [k]} |S^*_i|}{1+(\sum_{i \in [k]} |S^*_i|-1)(1-c(S^*_1,\ldots,S^*_k))}$-approximation algorithm for monotone MVSO($\F$) minimization.
\end{theorem}

In some situations the above result leads to approximation factors
highly preferable to those obtained for general functions, given the strong polynomial hardness that most of these problems present for objectives with curvature $1$. Examples of such situations include families
like $\F=\{V\}$, spanning trees, or perfect matchings, where exact algorithms are available for modular objectives (i.e. $\beta =1$
in those cases) and any optimal solution $(S^*_1,\ldots,S^*_k)$ satisfies $\sum_{i \in [k]} |S^*_i| = \Omega(n)$.
%Thus, in settings where the curvature is constant or order $1-\frac{1}{\log n}$, we go from polynomial approximation factors (for objectives with curvature $1$) to constant or logarithmic ones.
Thus, we go from polynomial approximation factors (for objectives with curvature $1$) to constant or logarithmic factors (for constant or order $1-\frac{1}{\log n}$ curvature).

Moreover, having the curvature $c(S^*_1,\ldots,S^*_k)$ can be much more beneficial than having the total curvature $c$.
For instance, for the problem
$ \min f(S_1,\ldots,S_k): S_1 \uplus \cdots \uplus S_k = V $
with $f(S_1,\ldots,S_k) = \min\{n, \sum_{i=1}^k |S_i|\}$. Here the total curvature of $f$ is $1$ (hence leading to an $n$-approximation in Theorem \ref{thm:curvature}), while the curvature $c(S^*_1,\ldots,S^*_k)$
with respect to any partition $(S^*_1,\ldots,S^*_k)$ is $0$ (and thus leading to an exact approximation via Theorem \ref{thm:curvature}).

In Section \ref{sec:hardness-monot} we give evidence that Theorem \ref{thm:curvature} is essentially tight, even for $\F=\{V\}$ where we show the following curvature dependent information-theoretic lower bound.

\begin{theorem}
	\label{thm:MV-monot-hardness-curvature}
	The monotone MVSO($\F$) minimization problem over $\F=\{V\}$ and objectives $f$ with total curvature $c$
	cannot be approximated to a ratio $o(\frac{n/\log n}{1+(\frac{n}{\log n}-1)(1-c)})$ in the value oracle model with polynomial number of queries.
\end{theorem}

Finally, we give an approximation in terms of the number of agents $k$, which may become preferable in settings
where $k$ is not too large.
\begin{theorem}
	\label{thm:min-k-gap}
	Suppose there is a (polytime) $\alpha(n)$-approximation for monotone SO($\F$) minimization
	based on rounding the convex relaxation. Then there is a (polytime) $k \alpha(n)$-approximation
	for monotone MVSO($\F$) minimization.
\end{theorem}

%We also consider the special class of \emph{agent symmetric} multivariate objectives, where we show...

\subsection{The multivariate model and applications}
\label{sec:applications}

Our second objective is to extend the multivariate model and show that in some cases
this larger class remains tractable. Specifically,
we define the
{\em capacitated multivariate submodular optimization (CMVSO) problem} as follows:
\begin{equation}
\label{mv}
\mbox{CMVSO$(\F)$}~~~~~~~
\begin{array}{rc}
\max / \min & f(S_1,S_2,\ldots,S_k) \\
\mbox{s.t.} & S_{1}\uplus \cdots\uplus S_{k}\in\F\\
& S_{i}\in\F_{i}\,,\,\forall i\in[k]
\end{array}
\end{equation}
where we are supplied with subfamilies $\mathcal{F}_i$.

Our results imply that one maintains good approximations even while adding interesting side  constraints.
For example, for a monotone maximization instance of CMVSO
where $\F$ is a $p$-matroid intersection and the $\F_i$ are all matroids, our results from Section \ref{sec:MASA}
lead to a $(\frac{1}{p+1} - \epsilon)$-approximation algorithm via the multilinear relaxation, or a $1/(p+2)$-approximation via a simple greedy algorithm.
We believe that these, combined with other results from Section \ref{sec:MASA}, substantially expand the family of tractable models (both in theory and practice) for maximization.

%To the best of our knowledge there has been no previous work in the maximization setting
%of Problem (\ref{ma}) for nontrivial collections $\F$ of subsets of $V$.
Many existing applications  fit into the CMVSO framework and  some of these can be enriched through the added flexibility of the  capacitated model.
For instance, one may include
set bounds on the variables: $L_{i}\subseteq S_{i}\subseteq U_{i}$
for each $i$, or simple cardinality constraints: $|S_{i}| \leq b_{i}$ for each $i$.
A well-studied (\cite{fleischer2006sap,goundan2007revisiting,calinescu2011maximizing}) application of CMVSO in the maximization setting is the Separable Assignment Problem (SAP), which corresponds to the setting where
the objective is separable and modular, the $\F_i$ are downward closed (i.e. hereditary) families, and $\F=2^V$.
The following example illustrates CMVSO's potential as a general model.
%We believe that this framework has considerable potential for future applications (particularly for maximization), an illustrate this next with a sensor placement problem where the {\sc CMVSO} model arises naturally.
%a concrete example in Section~\ref{sec:applications}.

\begin{example}[Sensor Placement with Multivariate Objectives]
	The problem of placing sensors and information gathering has been popular in the submodularity literature \cite{krause2007near,krause1973optimal,krause08efficient}.
	We are given a set of sensors $V$ and a set of possible locations $\{1,2,\ldots,k\}$ where the sensors can be placed.
	There is also a budget constraint restricting the total number of sensors that can be deployed.
	The goal is to place sensors at some of the locations so as to maximize the ``informativeness'' gathered.
	%It is again natural to consider a diminishing return property for this information gathering objective.
	This application is well suited to a $k$-multi-submodular objective function $f(S_{1},...,S_{k})$ which
	measures the ``informativeness'' of placing sensors $S_i$ at location $i$. A natural mathematical formulation for this is given by
	\[
	\begin{array}{cc}
	\max & f(S_{1},S_2,...,S_{k})\\
	\mbox{s.t.} & S_1 \uplus S_2 \uplus \cdots \uplus S_k \in \F\\
	& S_{i} \in \F_i,
	\end{array}
	\]
	where $\F:=\{ S \subseteq V: |S| \leq b \}$ imposes the budget constraint and $\F_i$ gives additional modelling flexibility. 
	\iffalse
	for the types (or number) of sensors that can be placed at location $i$. For instance, we could impose
	a cardinality constraint $|S_i|\leq b_i$ by defining $\F_i=\{ S \subseteq V: |S| \leq b_i \}$. We could also impose that only sensors $V_i \subseteq V$ can be placed at location $i$ by taking $\F_i=\{ S \subseteq V: S\subseteq V_i \}$. We can also combine these constraints, e.g.,  $\F_i=\{ S \subseteq V_i: |S| \leq b_i \}$.
	\fi
	For instance, we could impose $\F_i=\{ S \subseteq V_i: |S| \leq b_i \}$ to constrain
	the types and number of sensors that can be placed at location $i$.
	Notice that in  these cases both $\F$ and the $\F_i$ are matroids and hence the algorithms from Section \ref{sec:reduction-properties} apply.
	One may form a multivariate objective by defining $f(S_1,S_2, \ldots ,S_n)= \sum_i f_i(S_i) - R(S_1,S_2, \ldots ,S_n)$
	where the $f_i$'s measure the benefit of placing sensors $S_i$ at location $i$, and $R()$ is a redundancy function. If the $f_i$'s are submodular
	and $R()$ is k-multi-supermodular, then $f$ is $k$-multi-submodular.
	In this setting, it is natural to take the $f_i$'s to be coverage functions (i.e. $f_i(S_i)$ measures the coverage of placing sensors $S_i$ at location $i$).
	We next propose a family of   ``redundancy'' functions which are k-multi-supermodular.
	
	{\sc Supermodular penalty measures via Quadratic functions.}
	We denote  ${\bf S}:=(S_1,S_2, \ldots ,S_n)$ and
	define $z_{{\bf S}} := (|S_1|,|S_2|, \ldots ,|S_n|)$.
	One can  show (see Lemma~\ref{lem:MIMO} in Appendix \ref{sec:examples-mv-functions}) that if $A$ is a matrix satisfying $a_{ij}+a_{ji}\geq 0$, then $R({\bf S}):=z_{\bf S}^T A z_{\bf S}$ is k-multi-supermodular. Then for this particular example one could for instance take redundancy coefficients $a_{ij}$ as $\Theta(\frac{1}{d(i,j)^2})$ where $d(i,j)$ denotes the distance between locations $i$ and $j$.  This can be further extended  so that different sensor types
	contribute different weights to the vector $z_{{\bf S}}$, e.g., define $z_{{\bf S}}(i) = \sum_{j \in S_i} w(j)$ for an associated sensor weight vector $w$.
\end{example}

\iffalse
We now discuss Problem (\ref{mv}) in the \underline{minimization} setting.

\begin{example}[Rings]
	It is known \cite{schrijver2000combinatorial} that arbitrary
	submodular functions can be minimized efficiently
	on a ring family, i.e. a family of sets closed under unions and intersections
	(see \cite{orlin2009faster} for faster implementations, and work on modified ring families \cite{goemans2010symmetric}).
	A natural extension  is to  minimization
	of  $k$-multi-submodular functions over a \emph{multivariate ring family,
	}where by the latter we mean a family of tuples closed under (component
	wise) unions and intersections. We provide the formal definitions in
	Section \ref{sec:reduction-properties} and show that this
	more general problem can still be solved efficiently by applying
	the reduction from Section \ref{sec:MASA}.
\end{example}
\fi

\subsection{Related work}
\label{sec:related-work}

Submodularity naturally arises in many machine learning applications such as 
viral marketing \cite{kempe2003maximizing},
information gathering \cite{krause2007near}, image segmentation \cite{boykov2001interactive,kohli2009p3,jegelka2011submodularity},
document summarization \cite{lin2011class}, news article recommendation \cite{el2009turning}, active learning \cite{golovin2011adaptive}, and speeding up satisfiability
solvers \cite{streeter2009online}.

{\bf Single Agent Optimization.} The high level view of the tractability status for unconstrained (i.e., $\F=2^V$)
submodular optimization is that   both maximization and minimization generally behave well. Minimizing a submodular
set function is a classical combinatorial optimization problem which
can be solved in polytime \cite{grotschel2012geometric,schrijver2000combinatorial,iwata2001combinatorial}.
%,iwata2003faster
%,grotschel2012geometric
Unconstrained maximization, on the
other hand, is known to be inapproximable for general submodular set functions
%FBS, Remember what was the easy reduction.
but admits a polytime constant-factor approximation algorithm
when $f$ is nonnegative \cite{buchbinder2015tight,feige2011maximizing}.

In the constrained maximization setting, the classical work  \cite{nemhauser1978analysis,nemhauser1978best,fisher1978analysis} already
established an optimal $(1-1/e)$-approximation
factor for maximizing a nonnegative monotone submodular function subject to a
cardinality constraint, and a $(1/(k+1))$-approximation for maximizing
a nonnegative monotone submodular function subject to $k$ matroid constraints. This approximation is almost tight in the sense that there is an (almost matching)
factor $\Omega(\log(k)/k)$ inapproximability result \cite{hazan2006complexity}.
\iffalse
in classical work of Fisher, Nemhauser, and
Wolsey \cite{nemhauser1978analysis,nemhauser1978best,fisher1978analysis}
it is shown that the greedy algorithm achieves an optimal $(1-\frac{1}{e})$-approximation
factor for maximizing a monotone submodular function subject to a
cardinality constraint, and a $(k+1)$-approximation for maximizing
a monotone submodular function subject to $k$ matroid constraints. This approximation is almost tight in the sense that there is an (almost matching)
factor $\Omega(k/\log(k))$ inapproximability result.
Some of these results have been refined or improved recently.
\fi
For nonnegative monotone
functions, \cite{vondrak2008optimal,calinescu2011maximizing} give
an optimal $(1-1/e)$-approximation based on multilinear extensions when $\F$ is a matroid;
\cite{kulik2009maximizing} provides a
$(1-1/e-\epsilon)$-approximation when $\F$ is given by a constant number
of knapsack constraints, and \cite{lee2010submodular}
gives a local-search algorithm that achieves a $(1/k-\epsilon)$-approximation
(for any fixed $\epsilon>0$) when $\F$ is a $k$-matroid intersection. For nonnegative nonmonotone functions, a $0.385$-approximation is the best factor known \cite{buchbinder2016constrained} for maximization under a matroid constraint, in \cite{lee2009non} a $1/(k+O(1))$-approximation is given for $k$ matroid constraints with $k$ fixed. A simple ``multi-greedy'' algorithm \cite{gupta2010constrained} matches the approximation of Lee et al. but is polytime for any $k$. Vondrak \cite{vondrak2013symmetry} gives a $\frac{1}{2}(1-\frac{1}{\nu})$-approximation
under a matroid base constraint where $\nu$ denotes the fractional
base packing number. Finally, Chekuri et al \cite{vondrak2011submodular} introduce a general framework based on relaxation-and-rounding that allows for combining different types of constraints. This leads, for instance, to $0.38/k$ and $0.19/k$ approximations for maximizing nonnegative submodular monotone and nonmonotone functions respectively under the combination of $k$ matroids and $\ell=O(1)$ knapsacks constraints.
\iffalse
Recently, Orlin et al \cite{orlin2016robust} considered a robust formulation for constrained monotone submodular maximization, previously introduced by Krause et al \cite{krause2008robust}. The robustness in this model is with respect to the adversarial removal of up to $\tau$ elements. The problem can be stated as $\max_{S\in \F} \; \min_{A\subseteq S, |A|\leq \tau} \; f(S-A)$, where $\F$ is an independence system. Assuming an $\alpha$-approximation is available for the special case $\tau=0$ (i.e. the non-robust version), \cite{orlin2016robust} gives an $\alpha / (\tau +1)$-approximation for the robust problem.
\fi

For constrained minimization, the news is worse
\cite{goel2009approximability,svitkina2011submodular,iwata2009submodular}.
If $\F$ consists of spanning trees (bases of
a graphic matroid) Goel et al \cite{goel2009approximability} show a lower bound
of $\Omega(n)$, while in the case where $\F$ corresponds
to the cardinality constraint $\{S:|S|\geq k\}$ Svitkina and Fleischer
\cite{svitkina2011submodular} show a lower bound  of $\tilde{\Omega}(\sqrt{n})$.
There are a few exceptions. The problem can be solved exactly when $\F$ is a ring family (\cite{schrijver2000combinatorial}), triple family (\cite{grotschel2012geometric}), or parity family (\cite{goemans1995minimizing}).
In the context of NP-Hard problems, there are almost no cases where good (say $O(1)$ or $O(\log n)$) approximations exist.
%{\bf FBS. Cant the Tardos fac. location be
%viewed as a single-agent log-approx?. I guess not in a natural way. It requires looking at it as  b-matchings.}
We have that the submodular vertex cover admits
a $2$-approximation (\cite{goel2009approximability,iwata2009submodular}), and the $k$-uniform hitting
set has $O(k)$-approximation.

{\bf Multivariate Problems.}
The notion of $k$-multi-submodularity already appeared (under the name of multidimensional submodularity) in the classical work of Fisher et al \cite{fisher1978analysis}, where they consider the multivariate monotone maximization problem with $\F=\{V\}$ as a motivating example for submodular maximization subject to a matroid constraint. They show that for this problem a simple greedy algorithm achieves a $1/2$-approximation. The work of Singh et al \cite{singh2012bisubmodular} considers the special case of $2$-multi-submodular functions (they call them \emph{simple bisubmodular}). They give constant factor approximations for maximizing monotone $2$-multi-submodular functions under cardinality and partition constraints, and provide applications to coupled sensor placement and coupled feature selection problems.

Other different extensions of submodular functions to multivariate settings have been studied. Some of these include bisubmodular functions \cite{qi1988directed,ando1996characterization,fujishige2005bisubmodular,bouchet1995delta}, k-submodular functions \cite{huber2012towards,ward2014maximizing, yoshida2015monotone}, or skew bisubmodular functions \cite{huber2014skew,thapper2016complexity,thapper2012power}.

Finally, as mentioned in the Introduction, an important class of (multi-agent submodular optimization) problems arises
when $f \S = \sum_{i \in [k]} f_i(S_i)$. These problems have been widely studied in the case where $\F=\{V\}$,
both for minimization (\cite{hayrapetyan2005network,svitkina2010facility,ene2014hardness,chekuri2011submodular}) and maximization (\cite{fisher1978analysis,lehmann2001combinatorial,vondrak2008optimal}), and have also been considered for more general families \cite{goel2009approximability,santiago2018ma}.

%To the best of our knowledge, neither the MVSO($\F$) framework for general families $\F$ nor the notion of MV gap have been considered before in the literature.

%%%%%%%%%%%%%%%%%%%%%%%%%%%%%%%%%%%%%%%%%%%%%%%%%%%%%%%%%%%%%%%%%%%%%%%%%%%%%%%%%%%%%%%%%%%%%%%%%%%%%%%%%%%%%%%%%%%%%%%%%%%%%%%%%%%%%%%%%%%%%%%%%%%%%%%%%%%%%%
%%%%%%%%%%%%%%%%%%%%%%%%%%%%%%%%%%%%%%%%%%%%%%%%%%%%%%%%%%%%%%%%%%%%%%%%%%%%%%%%%%%%%%%%%%%%%%%%%%%%%%%%%%%%%%%%%%%%%%%%%%%%%%%%%%%%%%%%%%%%%%%%%%%%%%%%%%%%%%
\section{Multivariate submodular maximization}
\label{sec:MASA}
%%%%%%%%%%%%%%%%%%%%%%%%%%%%%%%%%%%%%%%%%%%%%%%%%%%%%%%%%%%%%%%%%%%%%%%%%%%%%%%%%%%%%%%%%%%%%%%%%%%%%%%%%%%%%%%%%%%%%%%%%%%%%%%%%%%%%%%%%%%%%%%%%%%%%%%%%%%%%%
%%%%%%%%%%%%%%%%%%%%%%%%%%%%%%%%%%%%%%%%%%%%%%%%%%%%%%%%%%%%%%%%%%%%%%%%%%%%%%%%%%%%%%%%%%%%%%%%%%%%%%%%%%%%%%%%%%%%%%%%%%%%%%%%%%%%%%%%%%%%%%%%%%%%%%%%%%%%%%

We describe two different reductions.
The first one reduces the capacitated multivariate problem CMVSO to a single-agent SO problem,
and it is based on the simple idea of taking $k$ disjoint copies of the original ground set.
We use this to establish an (optimal) MV gap of 1 for families such as spanning trees, matroids, and $p$-systems.
The second reduction is based on the multilinear extension of a set function.
We show that if the single-agent problem admits approximation via its multilinear relaxation (see Section \ref{sec:max-SA-MA-formulations}),
then we may extend this to its multivariate version with a constant factor loss, in the monotone and nonmonotone settings.
For the monotone case the MV gap is tight.

\subsection{The lifting reduction}
\label{sec:lifting-reduction}

We describe a generic reduction of CMVSO to a single-agent SO problem
$$\max/\min \bar{f}(S): S\in\L.$$
%This reduction, i.e. the definition of $\bar{f}$ and $\L$, is independent of whether we are considering a maximization or minimization  problem.
The argument is based on the idea of viewing assignments of elements $v$ to agents $i$ in a {\em multi-agent bipartite graph}.
This simple idea (which is equivalent to making $k$ disjoint copies of the ground set) already appeared in the classical work of Fisher et al \cite{fisher1978analysis}, and has since then been widely used \cite{lehmann2001combinatorial,vondrak2008optimal,calinescu2011maximizing,singh2012bisubmodular,santiago2018ma}.
We review briefly the reduction here for completeness and to fix notation. %that is used on later sections.

\iffalse
The argument is based on an observation of Lehmann et al. in \cite{lehmann2001combinatorial},
where they point out that the Submodular Welfare Problem is a special case of submodular
maximization subject to a matroid constraint.
This reduction has already been used in \cite{santiago2018ma} for the simpler case where the objective
has the form $f \S = \sum_{i \in [k]} f_i(S_i)$. Here we extend it to the more general setting of $k$-multi-submodular functions.
\fi

Consider the complete bipartite graph $G=([k]+V,E)$.
%We think of $E$ as our new ground set (or lifted space), where an edge $(i,v)\in E$ corresponds to assigning item $v$ to agent $i$.
Every subset of edges $S \subseteq E$ can be written uniquely as
$S = \uplus_{i \in [k]} (\{i\} \times S_i)$ for some sets $S_i \subseteq V$.
This allows us to go from a multivariate objective (such as the one in (\ref{mv})) to a univariate objective $\bar{f}:2^{E}\to\R$ over the lifted space. Namely, for each set $S \subseteq E$ we define $\bar{f}(S)=f(S_1,S_2,\ldots,S_k)$. The function $\bar{f}$ is well-defined because of the one-to-one correspondence between sets $S\subseteq E$ and tuples $\S \subseteq V^k$.

%We next give a description of the feasible sets in the lifted space $E$.
We consider two families of sets over $E$ that capture the original constraints:
$$
\F':=\{S\subseteq E:S_{1}\uplus \cdots\uplus  S_{k}\in\F\} \hspace{15pt} \mbox{and} \hspace{15pt}
\H:=\{S\subseteq E:S_{i}\in\F_{i},\;\forall i\in[k]\}.
$$
%Notice that in the definitions of $\F'$ and $\H$ we are again strongly using the fact that each $S\subseteq E$ can be written uniquely as $S=\uplus_{i\in[k]}(\{i\}\times S_{i})$ where $S_{i}\subseteq V$.

\noindent
We now have:
\[
\begin{array}{cccccccc}
\max/\min & f(S_1,S_2,\ldots,S_k) & = & \max/\min & \bar{f}(S) & = & \max/\min & \bar{f}(S)\\
\mbox{s.t.} & S_{1}\uplus \cdots\uplus S_{k}\in\F &  & \mbox{s.t.} & S\in\F' \cap \H &  & \mbox{s.t.} & S\in\L\\
& S_{i}\in\F_{i}\,,\,\forall i\in[k]
\end{array},
\]
where in the last step we just let $\L:=\F' \cap\H$.

\iffalse
Notice that for the robust problem discussed in Theorem \ref{thm:max robust} we get the single-agent problem
\begin{equation*}
\begin{array}{cccccccc}
\max & \min & \sum_{i \in [k]} f_i(S_i - A_i) & = & \max & \min & f(S-A). \\
S_1 \uplus  \cdots \uplus  S_k \in \F &  A_i \subseteq S_i & & & S \in \L & A \subseteq S& \\
S_{i} \in \F_{i} & \sum_i |A_i|\leq \tau & & & & |A| \leq \tau &
\end{array}
\end{equation*}
\fi

Clearly, this reduction is  interesting if our new function
$\bar{f}$ and the family of sets $\L$ have properties which allow us to handle
them computationally. This depends on the original structure
of the function $f$, and the set families $\F$ and $\F_{i}$. The following is straightforward.
\iffalse
In terms of the objective, it is straightforward to check that if $f$ is a (nonnegative, respectively monotone) $k$-multi-submodular
submodular function, then $\bar{f}$ as defined above is also (nonnegative, respectively monotone) submodular.
\fi
%In terms of the objective we have the following result.
\begin{claim}
	\label{claim:subm-invariance}
	If $f$ is a (nonnegative, respectively monotone) $k$-multi-submodular
	function, then $\bar{f}$ as defined above is also (nonnegative, respectively monotone) submodular.
\end{claim}

In Section \ref{sec:reduction-properties} we discuss several properties of the families $\F$ and $\F_i$ that are preserved under this reduction, as well as their algorithmic consequences.

\subsection{Multilinear extensions for MV problems}
\label{sec:max-SA-MA-formulations}

Given a set function $f:2^V \to \R$ (or equivalently $f:\{0,1\}^n \to \R$), we say that $g:[0,1]^n \to \R$ is an {\em extension} of $f$ if $g(\chi^S) = f(S)$ for each $S \subseteq V$. Clearly, there are many possible extensions that one could consider for any given set function. One that has been very useful in the submodular maximization setting due to its nice properties is the {\em multilinear extension}.

For a set function $f:\{0,1\}^V \to \R$ we define its \emph{multilinear extension}  $f^M:[0,1]^V \to \R$ (introduced in \cite{calinescu2007maximizing}) as
\begin{equation*}
f^M(z)=\sum_{S \subseteq V} f(S) \prod_{v \in S} z_v \prod_{v \notin S} (1-z_v).
\end{equation*}
%The above definition has a natural interpretation
An alternative way to define $f^M$ is in terms of expectations. Given a vector $z \in [0,1]^V$ let $R^z$ denote a random set that contains element $v_i$ independently with probability $z_{v_i}$. Then $f^M(z)= \E[f(R^z)]$, where the expectation is taken over random sets generated from the probability distribution induced by $z$. One very useful property of the multilinear extension is the following.

\begin{proposition}
\label{prop:multilinear-properties}
	Let $f:2^V \to \R$ be a submodular function and $f^M:[0,1]^n \to \R$ its multilinear extension.
	\iffalse
	Then we have the following.
	\begin{itemize}
		\item If $f$ is monotone then $F$ is monotone along any direction $d \geq 0$.
		
		\item If $f$ is submodular then $F$ is concave along any direction $d \geq 0$. In addition, $F$ is convex along any direction $d = \bf{e_i} - \bf{e_j}$ for $i,j \in \{1,2,\ldots,n\}$, where $\bf{e_i}$ denotes the vector with $ith$ component equals to one and the rest of components equal to zero.
	\end{itemize}
	\fi
	Then $f^M$ is convex along any direction $d = \bf{e_{v_i}} - \bf{e_{v_j}}$ for $i,j \in \{1,2,\ldots,n\}$, where $\bf{e_v}$ denotes the characteristic vector of $\{v\}$, i.e. the vector in $\R^V$ which has value $1$ in the $v$-th component and zero elsewhere.
\end{proposition}

This now gives rise to natural single-agent and multivariate relaxations.
The {\em single-agent multilinear extension relaxation} is:
\begin{equation}
\label{SA-ME}
(\mbox{SA-ME}) ~~~~
\max f^M(z): z \in P^*(\F),
\end{equation}
and the {\em multivariate multilinear extension relaxation} is:
\begin{equation}
\label{MV-ME}
(\mbox{MV-ME}) ~~~~
\max \bar{f}^M(z_1,z_2,\ldots,z_k): z_1 + z_2 + \cdots + z_k \in P^*(\F),
\end{equation}
where $P^*(\F)$ denotes some relaxation of the polytope $conv(\{\chi^S:S\in \F\})$ 
\footnote{$conv(X)$ denotes the convex hull of a set $X$ of vectors, and $\chi^S$
	denotes the characteristic vector of the set $S$.
}
, and $\bar{f}$ the lifted univariate function from the reduction in Section \ref{sec:lifting-reduction}.
Note that $\bar{f}$ is defined over vectors $\bar{z} = (z_1,z_2,\ldots,z_k) \in [0,1]^E$, where we think of $z_i \in \R^n$ as the vector associated to agent $i$.

The relaxation SA-ME has been used extensively \cite{calinescu2011maximizing,lee2009non,feldman2011unified,ene2016constrained,buchbinder2016constrained} in the submodular maximization literature.
%We are not aware of any previous use of the multivariate relaxation (MV-ME).
The following result shows that when $f$ is nonnegative submodular and the relaxation $P^*(\F)$ is downwards closed and admits a polytime separation oracle, the relaxation SA-ME can be solved approximately in polytime.

\begin{theorem}[\cite{buchbinder2016constrained,vondrak2008optimal}]
	\label{thm:multilinear-solve-monot}
	Let $f:2^V \to \R_+$ be a nonnegative submodular function and $f^M:[0,1]^V \to \R_+$ its multilinear extension. Let $P \subseteq [0,1]^V$ be any downwards closed polytope that admits a polytime separation oracle, and denote $OPT = \max f^M(z): z\in P$. Then there is a polytime algorithm (\cite{buchbinder2016constrained}) that finds $z^* \in P$ such that $f^M(z^*) \geq 0.385 \cdot OPT$. Moreover, if $f$ is monotone there is a polytime algorithm (\cite{vondrak2008optimal}) that finds $z^* \in P$ such that $f^M(z^*) \geq (1-1/e) OPT$.
\end{theorem}

For monotone objectives the assumption that $P$ is downwards closed is without loss of generality.
This is not the case, however, when the objective is nonmonotone. Nonetheless, this restriction
is unavoidable, as Vondr{\'a}k \cite{vondrak2013symmetry} showed that no algorithm can find $z^* \in P$ such that $f^M(z^*) \geq c \cdot OPT$
for any constant $c>0$ when $P$ admits a polytime separation oracle but it is not downwards closed.

\iffalse
\begin{theorem}[\cite{buchbinder2016constrained}]
	%\label{thm:multilinear-solve-non-monot}
	Let $f:2^V \to \R_+$ be a nonnegative submodular function and $f^M:[0,1]^V \to \R_+$ its multilinear extension. Let $P \subseteq [0,1]^V$ be any down-closed polytope that admits a polytime separation oracle. Then there is a polytime algorithm that finds $z^* \in P$ such that $f^M(z^*) \geq 0.385 \max f^M(z): z\in P$.
\end{theorem}
\fi

\iffalse
We show that if the single-agent (monotone or nonmonotone) formulation (SA-ME) can be solved approximately in polytime, then so can the corresponding (monotone or nonmonotone) multi-agent formulation (MA-ME) to the same approximation factor (see Appendix \ref{sec:max-extensions}).
\fi

We can solve the MV-ME relaxation to the same approximation factor that SA-ME. To see this note that the
multivariate problem has the form $\{ \max g(w) : w \in W \subseteq {\bf R}^{nk} \}$
where $W$ is the downwards closed polytope $\{w=(z_1,...,z_k): \sum_i z_i \in P^*(\F)\}$ and  $g(w)=\bar{f}^M(z_1,z_2,\ldots,z_k)$. Clearly we have a polytime separation oracle for $W$ given that we have one for $P^*(\F)$.
Moreover, $g$ is the multilinear extension of a nonnegative submodular function (since by Claim \ref{claim:subm-invariance} we know $\bar{f}$ is nonnegative submodular), and we can now use Theorem \ref{thm:multilinear-solve-monot}.
\iffalse
Moreover, it is straightforward to check (see Lemma \ref{lem:max-multilinear} on Appendix \ref{sec:Appendix-Invariance}) that $g(w)=f^M(w)=h^M(w)$, where $h$ is the function on the lifted space after applying the lifting reduction from Section \ref{sec:lifting-reduction}. Thus, $g$ is the multilinear extension of a nonnegative submodular function, and we can now use Theorem \ref{thm:multilinear-solve-monot}.
\fi
%by Lemma \ref{} we know that $g(w)=f^M(w)$, where $f$ is the function on the lifted space after applying the lifting reduction from Section \ref{sec:lifting-reduction}.

\iffalse
The multilinear extension $f^M$ of a submodular function $f$ has several nice properties.
One that will be particularly useful for us is the following.

\begin{lemma}[\cite{calinescu2007maximizing}]
	\label{lem:multilinear-convex}
	Let $f:2^V \to \R$ be a set function and $f^M:[0,1]^V \to \R$ its multilinear extension.
	If $f$ is submodular then $f^M$ is convex along any direction $d = \bf{e_{v_i}} - \bf{e_{v_j}}$ for $v_i,v_j \in V$, where $\bf{e_i}$ denotes the vector with $ith$ component equals to one and the rest of components equal to zero.
\end{lemma}
\fi

\subsection{A tight $1-1/e$ MV gap}
\label{sec:max-MA-gap}

In this section we prove Theorem \ref{thm:max-MA-gap}.
The main idea is that we start with an (approximate) optimal solution
$z^* = z_1^* + z_2^* + \cdots + z^*_k$ to the MV-ME relaxation and
build a new feasible solution $\hat{z} = \hat{z}_1 + \hat{z}_2 + \cdots + \hat{z}_k$
where the $\hat{z}_i$ have supports $V_i$ that are pairwise disjoint.
We think of $V_i$ as the set of items associated (or pre-assigned) to agent $i$.
%This induces a subpartition $V_1,V_2,\ldots,V_k$ of $V$ where $V_i$ are the items associated to agent $i$.
Once we have such a pre-assignment we consider the single-agent problem $\max g(S): S \in \F$ where
\begin{equation}
\label{g-function}
%\mbox{(NEW-SA)}~~~
g(S)= f(S \cap V_1,S \cap V_2,\ldots,S \cap V_k).
\end{equation}
It is clear that $g$ is nonnegative
monotone submodular since $f$ is nonnegative monotone $k$-multi-submodular.
Moreover, for any feasible solution $S \in \F$ for this single-agent problem, we obtain
a multivariate solution of the same cost by setting
$S_i = S \cap V_i$, since then
$g(S) = f(S \cap V_1,S \cap V_2,\ldots,S \cap V_k) = f(S_1,S_2,\ldots,S_k).$

For a set $S \subseteq V$ and a vector $z \in [0,1]^V $ we denote by $z|_S$ the truncation of $z$ to elements of $S$.
That is, we set $z|_S (v) = z(v)$ for each $v \in S$ and to zero otherwise.
Then by definition of $g$ we have that $g^M(z) = \bar{f}^M(z|_{V_1},z|_{V_2},\ldots,z|_{V_k})$,
where $\bar{f}$ is the lifted function from Section \ref{sec:lifting-reduction}.
Moreover, if the sets $V_i$ are pairwise
disjoint, then $\bar{f}^M(z|_{V_1},z|_{V_2},\ldots,z|_{V_k}) =\bar{f}^M(z_1,z_2,\ldots,z_k)$.
The next result formalizes this observation.

\begin{proposition}
	\label{prop:max-g-function}
	Let $z = \sum_{i\in [k]} z_i$ be a feasible solution to MV-ME such that the
	vectors $z_i$ have pairwise disjoint supports $V_i$. Then
	$g^M(z) = \bar{f}^M(z_1,z_2,\ldots,z_k).$
\end{proposition}

We now have all the ingredients to prove our main result for maximization.
We note that a gap of $1-1/e$ appeared in \cite{santiago2018ma} for the case of separable objectives
$f \S = \sum_i f_i(S_i)$. That argument uses the component-wise linearity of the 
multilinear extension, while our proof for non-separable objectives strongly uses 
the convexity property from Proposition \ref{prop:multilinear-properties}.

\begin{reptheorem}{thm:max-MA-gap}
	If there is a (polytime) $\alpha(n)$-approximation for monotone SO($\F$) maximization
	based on rounding SA-ME, then there is a (polytime) $(1-1/e) \cdot \alpha(n)$-approximation
	for monotone MVSO($\F$) maximization. Furthermore, given a downwards closed family $\F$,
	if there is a (polytime) $\alpha(n)$-approximation
	for nonmonotone SO($\F$) maximization
	based on rounding SA-ME, then there is a (polytime) $0.385 \cdot \alpha(n)$-approximation
	for nonmonotone MVSO($\F$) maximization.
\end{reptheorem}
\begin{proof}%[Theorem \ref{thm:max-MA-gap}]
	
	We  discuss first the case of monotone objectives.
	
	{\sc STEP 1.}	
	Let $z^* = z_1^* + z_2^* + \cdots + z^*_k$ denote an approximate solution to MV-ME obtained via Theorem \ref{thm:multilinear-solve-monot}, and let $OPT_{frac}$ be the value of an optimal solution. We then have that $f^M(z^*_1,z^*_2,\ldots,z^*_k) \geq (1-1/e) OPT_{frac} \geq (1-1/e) OPT_{MV}$.
	
	{\sc STEP 2.}
	For an element $v \in V$ let $\bf{e_v}$ denote the characteristic vector of $\{v\}$, i.e. the vector in $\R^V$ which has value $1$ in the $v$-th component and zero elsewhere.
	Then by Proposition \ref{prop:multilinear-properties} we have that the function
	\begin{equation*}
	h(t) = \bar{f}^M(z^*_1,z^*_2,\ldots,z^*_{i-1},z^*_i + t \mathbf{e_v},z^*_{i+1},\ldots,z^*_{i'-1},z^*_{i'} - t \mathbf{e_v},z^*_{i'+1},\ldots,z^*_k)
	\end{equation*}
	%f^M_i (z^*_i + t \mathbf{e_v} ) + f^M_{i'} (z^*_{i'} - t \mathbf{e_v} )  +  \sum_{j\in [k], j\neq i,i'} f^M_j(z^*_j)
	is convex for any $v\in V$ and $i \neq i' \in [k]$. In particular, given any $v \in V$ such that there exist $i \neq i' \in [k]$ with $z^*_i(v),z^*_{i'}(v)>0$, there is always a choice so that increasing one component and decreasing the other by the same amount does not decrease the objective value.	
	
	Let $v \in V$ be such that there exist $i \neq i' \in [k]$ with $z^*_i(v),z^*_{i'}(v)>0$. Then, we either set $z^*_i(v) = z^*_i(v) + z^*_{i'}(v)$ and $z^*_{i'}(v) = 0$, or $z^*_{i'}(v) = z^*_i(v) + z^*_{i'}(v)$ and $z^*_i(v) = 0$, whichever does not decrease the objective value. We repeat until the vectors $z^*_i$ have pairwise disjoint support. Let us denote these new vectors by $\hat{z}_i$ and let $\hat{z}= \sum_{i\in [k]} \hat{z}_i$. Then notice that the vector $z^* = \sum_{i\in [k]} z^*_i$ remains invariant after performing each of the above updates (i.e. $\hat{z} = z^*$), and hence the new vectors $\hat{z}_i$ remain a feasible solution.

	{\sc STEP 3.}
	In the last step we use the function $g$ defined in (\ref{g-function}), with sets $V_i$ corresponding to the supports of the $\hat{z}_i$.
	Given our $\alpha$-approximation rounding assumption for SA-ME, we can round $\hat{z}$
	to find a set $\hat{S}$ such that $g(\hat{S})\geq \alpha g^M(\hat{z})$.
	Then, by setting $\hat{S}_i = \hat{S} \cap V_i$ we obtain a multivariate solution satisfying
	\begin{equation*}
	f(\hat{S}_1,\ldots,\hat{S}_k) = g(\hat{S}) \geq \alpha g^M(\hat{z}) = \alpha f^M(\hat{z}_1,\ldots,\hat{z}_k) \geq \alpha f^M(z^*_1,\ldots,z^*_k) \geq \alpha (1-1/e) OPT_{MV},
	\end{equation*}
	where the second equality follows from Proposition \ref{prop:max-g-function}.
	%This completes the proof for monotone objectives.
	This completes the monotone proof.
	
	For the nonmonotone case the argument is very similar. %Notice that in the nonmonotone case we can only approximate the solution to the multilinear extension in the case where the family $\F$ is downwards closed, i.e. when $P(\F)$ is downwards closed. Thus,
	Here we restrict our attention to downwards closed families, since then we can get a $0.385$-approximation at STEP 1 via Theorem \ref{thm:multilinear-solve-monot}. We then apply STEP 2 and 3 in the same fashion as we did for monotone objectives. This leads to a $0.385 \cdot \alpha(n)$-approximation for the multivariate problem.
\end{proof}

\subsection{Invariance under the lifting reduction}
\label{sec:reduction-properties}

In Section \ref{sec:max-MA-gap} we established a MV gap of $1-1/e$ for monotone objectives and of $0.385$ for nonmonotone objectives and downwards closed families based on the multilinear formulations. In this section we describe several families with an (optimal) MV gap of $1$. Examples of such family classes include spanning trees, matroids, and $p$-systems. 
%Moreover, the reduction in this case is completely black box, and hence does not depend on the multilinear (or any other particular) formulation.

We saw in Section \ref{sec:lifting-reduction} that if the original function $f$ is $k$-multi-submodular then the lifted function $\bar{f}$ is submodular. We now discuss some properties of the original families $\F_i$ and $\F$ that are also preserved under the lifting reduction; these were already proved in \cite{santiago2018ma}. It is shown there, for instance, that if $\F$ induces a matroid (or more generally a $p$-system) over the ground set $V$, then so does the family $\F'$ over the lifted space $E$.
We summarize these results in Table \ref{table:properties-preserved}, and
discuss next some of the algorithmic consequences.

%\vspace*{-0.6cm}

\begin{table}[ht]
	\caption{Invariant properties under the lifting reduction}
	\label{table:properties-preserved}
	\resizebox{\linewidth}{!}{
		\begin{tabular}{|c|c|c|c|}
			\hline
			& \textbf{Multivariate problem} & \textbf{Single-agent (i.e. reduced) problem} & Result\tabularnewline
			\hline
			1 & $f$ $k$-multi-submodular & \mbox{ \footnotesize $\bar{f}$} submodular & Section \ref{sec:lifting-reduction} \tabularnewline
			\hline
			2 & $f$ monotone & \mbox{ \footnotesize $\bar{f}$} monotone & Section \ref{sec:lifting-reduction} \tabularnewline
			\hline
			3 & $(V,\F)$ a $p$-system & $(E,\F')$ a $p$-system & \cite{santiago2018ma} \tabularnewline
			\hline
			4 & $\F$ = bases of a $p$-system & $\F'$ = bases of a $p$-system & \cite{santiago2018ma} \tabularnewline
			\hline
			5 & $(V,\F)$ a matroid & $(E,\F')$ a matroid & \cite{santiago2018ma}\tabularnewline
			\hline
			6 & $\F$ = bases of a matroid & $\F'$ = bases of a matroid & \cite{santiago2018ma} \tabularnewline
			\hline
			7 & $(V,\F)$ a $p$-matroid intersection & $(E,\F')$ a $p$-matroid intersection & \cite{santiago2018ma} \tabularnewline
			\hline
			\iffalse
			8 & $\F=\{V\}$ & $\F'=$ set of bases of a partition matroid & \cite{santiago2018ma} \tabularnewline
			\hline
			9 & $\F=2^{V}$  & $(E,\F')$ a partition matroid & \cite{santiago2018ma} \tabularnewline
			\hline
			\fi
			8 & $(V,\F_{i})$ a matroid for all $i\in[k]$ & $(E,\H)$ a matroid & \cite{santiago2018ma} \tabularnewline
			\hline
			9 & $\F_{i}$ a ring family for all $i\in[k]$ & $\H$ a ring family & \cite{santiago2018ma} \tabularnewline
			\hline
			10 & $\F=$ forests (resp. spanning trees) & $\F'=$ forests (resp. spanning trees) &  \cite{santiago2018ma}\tabularnewline
			\hline
			11 & $\F=$ matchings (resp. perfect matchings) & $\F'=$ matchings (resp. perfect matchings) &  \cite{santiago2018ma}\tabularnewline
			\hline
			12 & $\F=$ $st$-paths & $\F'=$ $st$-paths &  \cite{santiago2018ma}\tabularnewline
			\hline
		\end{tabular}
	}
\end{table}

%\vspace*{-0.6cm}

In the setting of MVSO (i.e. (\ref{eqn:MV})) this invariance allows us to leverage several results
from the single-agent to the multivariate setting. These are based on the following result, which uses
the fact that the size of the lifted space $E$ is $nk$.

\begin{theorem}
	\label{thm:max-invariance1}
	Let $\F$ be a matroid, a $p$-matroid intersection, or a $p$-system. If there is a (polytime) $\alpha(n)$-approximation algorithm for monotone (resp. nonmonotone) SO($\F$) maximization (resp. minimization),
	then there is a (polytime) $\alpha(nk)$-approximation algorithm for monotone (resp. nonmonotone) MVSO($\F$) maximization (resp. minimization).
\end{theorem}

For both monotone and nonmonotone maximization the approximation factors $\alpha(n)$ for the family classes described in Theorem \ref{thm:max-invariance1} are independent of (the size of the ground set) $n$. Hence, we immediately get that
$\alpha(nk)=\alpha(n)$ for those cases, and thus approximation factors for the corresponding multivariate and single-agent problems are the same.
In our MV gap terminology this implies an MV gap of 1 for such problems. This proves Theorem \ref{thm:max-invariance}.

In the setting of CMVSO (i.e. (\ref{mv})) the results described on entries $8$ and $9$ of Table \ref{table:properties-preserved}
provide additional modelling flexibility. This allows us to maintain decent approximations while combining several constraints.
For instance, for a monotone maximization instance of CMVSO
where $\F$ corresponds to a $p$-matroid intersection and the $\F_i$ are all matroids, the above invariance results
lead to a $(\frac{1}{p+1} - \epsilon)$-approximation.

The results from this section also imply that algorithms that behave very well in practice
(such as accelerated greedy variants \cite{mirzasoleiman2016fast} and distributed algorithms \cite{mirzasoleiman2016distributed}) 
%(e.g. lazy greedy \cite{minoux1978accelerated}, stochastic greedy \cite{mirzasoleiman2015lazier}, and distributed algorithms \cite{mirzasoleiman2016distributed}) 
for the corresponding single-agent problems, can also be used for the more general multivariate setting while preserving the same approximation guarantees. We believe this makes the CMVSO framework a good candidate for potential applications in large-scale machine learning problems.

\section{Multivariate submodular minimization}

In this section we present different approximation factors in terms of $n$ (i.e. the number of items) and $k$ (i.e. the number of agents) for the
monotone multivariate problem. Moreover, the approximation factors in terms of $n$ are essentially tight.

\iffalse
Our approximation results and MV gaps are much larger in the minimization side. However, this is somewhat expected due to the strong hardness results already existing for constrained single-agent submodular minimization (see Section \ref{sec:related-work}). 
\fi

\subsection{A $\frac{\beta \cdot n}{1+(n-1)(1-c)}$-approximation}
\label{sec:min-curvature}
Let $f_S(v)=f(S+v)-f(S)$ denote the marginal gain of adding $v$ to $S$.
Given a normalized monotone submodular function $f:2^V \to \R$, its \emph{total curvature} $c$ and its \emph{curvature $c(S)$ with respect to a set} $S\subseteq V$ are defined as (in \cite{conforti1984submodular,vondrak2010curvature})
\begin{equation*}
c=\max_{j\in V} \frac{f(j)- f_{V-j}(j)}{f(j)} = 1 - \min_{j \in V} \frac{f_{V-j}(j)}{f(j)} \hspace*{0.5cm}\mbox{and} \hspace*{0.5cm} c(S)= 1 - \min_{j \in S} \frac{f_{S-j}(j)}{f(j)}.
\end{equation*}
We may think of this number as indicating how far the function $f$ is from being modular (with $c=0$ corresponding to being modular).
The notion of curvature has been widely used for univariate functions \cite{conforti1984submodular,vondrak2010curvature,iyer2013curvature,bai2018greed}, since it allows for better approximations and it is linear time computable.

Given a tuple $(S_1,\ldots,S_k) \in 2^{kV}$ and $(i,v) \in [k]\times V$, we denote by $(S_1,\ldots,S_k)+(i,v)$ the new tuple $(S_1,\ldots,S_{i-1},S_i + v, S_{i+1},\ldots,S_k)$.
It is natural to think of the quantity
\begin{equation*}
\label{eq:mv-diminishing-returns}
f_{\S}((i,v)) := f(\S + (i,v)) - f \S
\end{equation*}
as the marginal gain of assigning element $v$ to agent $i$ in the tuple $(S_1,\ldots,S_k)$.
%We use the notation $f_{\S}((i,v))$ to denote the marginal gain in (\ref{eq:mv-diminishing-returns}).
We also use $f((i,v))$ to denote the quantity $f(\emptyset,\ldots,\emptyset,v,\emptyset,\ldots,\emptyset)$ where $v$ appears in the $ith$ component.

%Similarly to the univariate case, 
Given a normalized monotone $k$-multi-submodular function $f:2^{kV}\to \R$ we define its \emph{total curvature} $c$ and its \emph{curvature $c(S_1,\ldots,S_k)$ with respect to a tuple} $(S_1,\ldots,S_k) \subseteq V^k$ as
\begin{equation*}
c = 1 - \min_{i\in [k], v\in V} \frac{f_{(V,V,\ldots,V)-(i,v)}((i,v))}{f((i,v))} \hspace*{0.3cm}, \hspace*{0.3cm} c \S = 1 - \min_{ i\in [k], v \in S_i} \frac{f_{\S - (i,v)}((i,v))}{f((i,v))}
\end{equation*}

There is a straightforward correspondence between the curvature of $f$ and the curvature of its lifted version $\bar{f}$.

\iffalse
We note that the above definition is equivalent to say that $f$ has curvature $c$ if and only if its corresponding univariate function $\bar{f}$ on the lifted space has curvature $c$ (in the standard univariate sense).
\fi

\begin{claim}
	\label{claim:curvature-lifted-function}
	Let $f:2^{kV}\to \R_+$ be a normalized nonnegative monotone $k$-multi-submodular function,
	and $\bar{f}:2^E \to \R_+$ the corresponding lifted function. Then, $f$ has total curvature
	$c$ if and only if $\bar{f}$ has total curvature $c$. Also, $f$ has curvature  $c \S$ with
	respect to a tuple if and only if $\bar{f}$ has curvature $c(S)$ with respect to the set $S$ in the
	lifted space corresponding to the tuple $\S$.
\end{claim}

The following curvature dependent result for univariate functions was proved in \cite{iyer2013curvature}.%by Iyer et al
%We use the following result for univariate functions from Iyer et al \cite{iyer2013curvature}.
\begin{proposition}[\cite{iyer2013curvature}]
	\label{prop:curvature-bilmes}
	Let $f:2^V \to \R$ be a nonnegative monotone submodular function, and $w:V\to \R_+$ the modular
	function given by $w(v)=f(v)$. Let $c(S)$ denote the curvature of $f$ with respect to $S$,	
	and $S^*$ denote an optimal solution to  $\min f(S): S\in \F$.
	Let $\hat{S} \in \F$ be a $\beta$-approximation for the problem $\min w(S):S \in \F$.
	Then
	$$
	f(\hat{S}) \leq \frac{\beta |S^*|}{1+(|S^*| -1)(1-c(S^*))} f(S^*).
	$$	
\end{proposition}

We extend the above result to the setting of $k$-multi-submodular objectives.
\begin{reptheorem}{thm:curvature}
	Let $f$ be a nonnegative monotone $k$-multi-submodular function, and let $\F$ be a family that admits a (polytime) $\beta$-approximation over modular functions. Denote by $(S^*_1,\ldots,S^*_k)$ an optimal solution to monotone MVSO($\F$) minimization, and by $c(S^*_1,\ldots,S^*_k)$ the curvature of $f$ with respect to $(S^*_1,\ldots,S^*_k)$. Then there is a (polytime) $\frac{\beta \sum_{i \in [k]} |S^*_i|}{1+(\sum_{i \in [k]} |S^*_i|-1)(1-c(S^*_1,\ldots,S^*_k))}$-approximation algorithm for monotone MVSO($\F$) minimization.
\end{reptheorem}
\begin{proof}
	Let $\bar{f}:2^E \to \R_+$ and $\F'$ be the lifted function and family described in the lifting reduction from Section \ref{sec:lifting-reduction}. We then have
	\[
	\begin{array}{cccccccc}
	\min & f(S_1,S_2,\ldots,S_k) & = &\min & \bar{f}(S)\\
	\mbox{s.t.} & S_{1}\uplus \cdots\uplus S_{k}\in\F &  & \mbox{s.t.} & S\in\F'
	\end{array}.
	\]
	
	Define a modular function $\bar{w}:E \to \R_+$ over the edges of the bipartite graph by
	$\bar{w}(i,v) = \bar{f}(i,v)$. Also, let $OPT = \min \bar{f}(S):S \in \F'$ and denote by $S^*$
	such a minimizer. Then by Proposition \ref{prop:curvature-bilmes} we have that any $\beta$-approximation
	for the modular minimization problem $\min \bar{w}(S): S \in \F'$ is a
	$\beta |S^*| / ( 1+(|S^*| -1)(1-c(S^*)))$-approximation for the problem $\min \bar{f}(S): S\in \F'$
	(and hence also for our original multivariate problem).
	Moreover, notice that by Claim \ref{claim:curvature-lifted-function} the curvature $c(S^*)$ of $\bar{f}$ with respect to $S^*$ is the same
	as the curvature $c(S^*_1,\ldots,S^*_k)$ of $f$ with respect to $(S^*_1,\ldots,S^*_k)$, where
	$(S^*_1,\ldots,S^*_k)$ is the tuple associated to $S^*$. Thus, we immediately get the desired approximation
	assuming that a (polytime) $\beta$-approximation is available for $\min \bar{w}(S): S \in \F'$.
	
	However, the lifted family $\F'$ could be more complicated than $\F$,
	and hence we would like to have an assumption depending on the original $\F$ (and not $\F'$).
	This can be achieved at no extra loss using the modularity of $\bar{w}$.
	Indeed, we can define a new modular function $w:V\to \R_+$ as $w(v) = \argmin_{i \in [k]} \bar{w}(i,v)$
	for each $v\in V$, breaking ties arbitrarily. It is then clear that
	$\min \bar{w}(S): S \in \F' = \min w(S): S \in \F$, since $\bar{w}$ always uses the cheapest
	copy of $v$ (i.e. assign $v$ to the agent $i$ with the smallest cost for it).
	
	We then get that any $\beta$-approximation for $\min w(S): S \in \F$ is also a $\beta$-approximation for
	$\min \bar{w}(S): S \in \F'$, and hence the desired result in terms of $\F$ follows.
\end{proof}

\iffalse
% Moved this to the Intro
Having the curvature $c(S^*_1,\ldots,S^*_k)$ with respect to the tuple $(S^*_1,\ldots,S^*_k)$
can make a big difference versus having the total curvature $c$ instead. A simple example is the problem
$$
\min f(S_1,\ldots,S_k): S_1 \uplus \cdots \uplus S_k = V
$$
with $f(S_1,\ldots,S_k) = \min\{n, \sum_{i=1}^k |S_i|\}$. Clearly the total curvature of $f$ is $1$
(hence leading to an $n$-approximation in Theorem \ref{thm:curvature}), but the curvature $c(S^*_1,\ldots,S^*_k)$
with respect to any partition $(S^*_1,\ldots,S^*_k)$ is $0$ (and thus leading to an exact approximation via Theorem \ref{thm:curvature}).
\fi

%\subsection{The single-agent and multivariate formulations}
%\label{sec:SA-MV-formulations}

\subsection{MV gap of $k$}
\label{sec:k-proof}

Due to monotonicity, one may often assume that we are working with a
family $\F$ which is {\em upwards-closed}, aka a {\em blocking family}.
%\footnote{
This can be done without loss of generality even if we seek polytime algorithms, since separation over a polytope with vertices $\{\chi^F: F \in \F\}$ implies separation over its dominant. We refer the reader to Appendix~\ref{sec:blocking} for details.
%}
\iffalse
The advantage is that to certify whether $F \in \F$, we only need to check that $F \cap B \neq \emptyset$ for each element $B$ of the family $\B(\F)$ of minimal blockers of $\F$. We discuss the details in Appendix \ref{sec:blocking}.
\fi

For a normalized set function $f:\{0,1\}^V \to \R$ one can define
its {\em Lov\'asz extension} $f^L:\R_+^V \to \R$ (introduced
in \cite{lovasz1983submodular}) as follows.
Let  $0   <  v_1 < v_2 < ... < v_m$  be the distinct
positive values taken in some vector $z \in \R_+^V$, and let $v_0=0$.
For each $i \in \{0,1,...,m\}$ define the set $S_i:=\{ j:  z_j > v_i\}$. In particular, $S_0$ is the support of $z$
and  $S_m=\emptyset$. One then defines:
\[
f^L(z) =  \sum_{i=0}^{m-1}   (v_{i+1}-v_i) f(S_i).
\]

It follows from the definition that $f^L$ is positively homogeneous, that is $f^L(\alpha z)=\alpha f^L(z)$ for any $\alpha > 0$ and $z \in \R_+^V$.
Moreover, it is also straightforward to see that $f^L$ is a monotone function if $f$ is.
%We use both of these properties of $f^L$ in our proofs.
We have the following result due to Lov\'asz.
\begin{lemma}
	\label{lem:LE-convex}
	[Lov\'asz \cite{lovasz1983submodular}]
	The function $f^L$ is convex if and only if $f$ is submodular.
\end{lemma}

This now gives rise to natural convex relaxations for the single-agent and multivariate problems based on some fractional relaxation $P^*(\F)$ of the integral polyhedron $conv(\{\chi^S:S\in \F\})$.
\iffalse
Tractability of these relaxations is discussed in detail in Appendix \ref{sec:relaxations}.
\fi
The {\em single-agent Lov\'asz extension formulation} (used in \cite{iwata2009submodular,iyer2014monotone}) is:
\begin{equation}
\label{SA-LE}
(\mbox{SA-LE}) ~~~~
\min f^L(z): z \in P^*(\F),
\end{equation}
and the {\em multivariate Lov\'asz extension formulation} is:
\begin{equation}
\label{MV-LE}
(\mbox{MV-LE}) ~~~~
\min \bar{f}^L(z_1,z_2,\ldots,z_k): z_1 + z_2 + \cdots + z_k \in P^*(\F),
\end{equation}
where $\bar{f}$ is the lifted univariate function from Section \ref{sec:lifting-reduction}.

\iffalse
The relaxation (SA-LE) has been already considered (\cite{iwata2009submodular,iyer2014monotone}) for different types of families $\F$, while we are only aware of (MA-LE) being used (\cite{chekuri2011submodular}) in the case $\F=\{V\}$.
\fi

By standard methods (e.g. see \cite{santiago2018ma}) one may solve SA-LE in polytime
if one can separate over the relaxation $P^*(\F)$. This is often the case for many
natural families such as spanning trees, perfect matchings, $st$-paths, and vertex covers.

We can also solve MV-LE as long as we have polytime separation of $P^*(\F)$.
This follows from the fact that the multivariate problem has the form $\{ \min g(w) : w \in W \subseteq {\bf R}^{nk} \}$
where $W$ is the full-dimensional convex body $\{w=(z_1,...,z_k): \sum_i z_i \in P^*(\F)\}$ and $g(w)=\bar{f}^L(z_1,z_2,\ldots,z_k)$.
Clearly we have a polytime separation oracle for $W$ given that we have one for $P^*(\F)$.
Moreover, by Lemma \ref{lem:LE-convex} and Claim \ref{claim:subm-invariance}, $g$ is convex since it is the Lov\'asz extension of a nonnegative submodular function $\bar{f}$.
Hence we may apply Ellipsoid as in the single-agent case.

%\subsection{A $k \alpha(n)$-approximation}
%\label{sec:k-proof}

We now give an approximation in terms of the number of agents, which becomes preferable when $k$ is not too large.
The high-level idea behind our reduction is the same as in the maximization setting (see Section \ref{sec:max-MA-gap}). That is, we start with an optimal solution $z^* = z_1^* + z_2^* + \cdots + z^*_k$ to the multivariate MV-LE relaxation and
build a new feasible solution $\hat{z} = \hat{z}_1 + \hat{z}_2 + \cdots + \hat{z}_k$
where the $\hat{z}_i$ have supports $V_i$ that are pairwise disjoint.
We then use for the rounding step the single-agent problem (as previously defined in (\ref{g-function}) for the maximization setting) $\min g(S): S \in \F$ where $g(S)= f(S \cap V_1,S \cap V_2,\ldots,S \cap V_k)$.

Similarly to Proposition \ref{prop:max-g-function} which dealt with the multilinear extension,
we have the following result for the Lov\'asz extension.

\begin{proposition}
	\label{prop:g-function}
	Let $z = z_1 + z_2 + \cdots + z_k$ be a feasible solution to MV-LE such that the
	vectors $z_i$ have pairwise disjoint supports $V_i$. Then
	$g^L(z) = \bar{f}^L(z_1,z_2,\ldots,z_k).$
\end{proposition}

\begin{reptheorem}{thm:min-k-gap}
	Suppose there is a (polytime) $\alpha(n)$-approximation for monotone SO($\F$) minimization
	based on rounding SA-LE. Then there is a (polytime) $k \alpha(n)$-approximation
	for monotone MVSO($\F$) minimization.
\end{reptheorem}
\begin{proof}
	Let $z^* = z_1^* + z_2^* + \cdots + z^*_k$ denote an optimal solution to MV-LE with value $OPT_{frac}$.
	We build a new feasible solution $\hat{z} = \hat{z}_1 + \hat{z}_2 + \cdots + \hat{z}_k$ as follows.
	For each element $v \in V$ let $i' = \argmax_{i \in [k]} z^*_i(v)$, breaking ties arbitrarily.
	Then set $\hat{z}_{i'}(v)=k z^*_i(v)$ and $\hat{z}_{i}(v)=0$ for each $i\neq i'$.
	By construction we have $\hat{z} \geq z^*$, and hence this is indeed a feasible solution. Moreover, by construction we also
	have that $\hat{z}_i \leq k z_i^*$ for each $i \in [k]$. Hence, given the monotonicity and homogeneity of $\bar{f}^L$ it follows that	
	\begin{equation*}
	\bar{f}^L(\hat{z}_1,\hat{z}_2,\ldots,\hat{z}_k) \leq \bar{f}^L(kz^*_1,kz^*_2,\ldots,kz^*_k) = k \bar{f}^L(z^*_1,z^*_2,\ldots,z^*_k) = k \cdot OPT_{frac} \leq k \cdot OPT_{MV}.
	\end{equation*}
	Since the $\hat{z}_i$ have disjoint supports $V_i$,
	for the single-agent rounding step we can now use the function $g$ defined in (\ref{g-function}) with the sets $V_i$.
	Given our $\alpha$-approximation rounding assumption for SA-LE, we can round $\hat{z}$
	to find a set $\hat{S}$ such that $g(\hat{S})\leq \alpha g^L(\hat{z})$.
	Then, by setting $\hat{S}_i = \hat{S} \cap V_i$ we obtain a multivariate solution satisfying
	\begin{equation*}
	f(\hat{S}_1,\hat{S}_2,\ldots,\hat{S}_k) = g(\hat{S}) \leq \alpha g^L(\hat{z}) = \alpha \bar{f}^L(\hat{z}_1,\hat{z}_2,\ldots,\hat{z}_k)
	\leq \alpha k \cdot OPT_{MV},
	%\leq \alpha \cdot O( \log (n) \log (\frac{n}{\log n}) ) \cdot OPT_{MA},
	\end{equation*}
	where the second equality follows from Proposition \ref{prop:g-function}.
	This completes the proof.
\end{proof}

The above theorem has interesting consequences. We now discuss one that leads to a polytime $k$-approximation for a much more general version of the submodular facility location problem considered by Sviktina and Tardos \cite{svitkina2010facility}, where $k$ denotes the number of facilities.

\begin{corollary}
	\label{Cor Non-separable}
	There is a  polytime $k$-approximation for the monotone MVSO($\F$) minimization problem over $\F=\{V\}$.
\end{corollary}
\begin{proof}
	Notice that the single-agent version of the above multivariate problem is the trivial $\min f(S): S \in \{V\}$. Hence a polytime exact algorithm is available for the single-agent problem and thus a polytime $k$-approximation is available for the multivariate version.
\end{proof}

\iffalse
A very special case of Corollary \ref{Cor Non-separable} occurs when
$f(S_{1},...,S_{k})=\sum_{i\in[k]}f_{i}(S_{i})$ for some nonnegative
monotone submodular $f_{i}$'s. This is equivalent to the submodular facility location
problem considered by Sviktina and Tardos in \cite{svitkina2010facility}, where they give
a  tight $O(\log n)$ approximation. Their approximation
is in terms of $n$, which in this setting denotes the number of customers/clients/demands.
Notice that from Corollary \ref{Cor Non-separable} we can immediately provide an
approximation in terms of the number of facilities. This bound clearly becomes preferable in
facility location problems (for instance, Amazon) where the number of customers swamps the number
of facility locations  (i.e. $n >> k$).

\begin{corollary}
	\label{Cor facility-location}
	There is a polytime $k$-approximation for submodular facility location, where
	$k$ denotes the number of facilities.
\end{corollary}
\fi

\subsection{An $o(\frac{n}{\log n})$ lower bound hardness for $\F=\{V\}$}
\label{sec:hardness-monot}

In this section we focus on the special case where $\F=\{V\}$. % the family of feasible sets consists only of $V$.
That is, we are looking for the optimal splitting of all the elements among the agents.
\iffalse
We remark that the $k$-multi-submodular objective generalizes its MA counterpart,
which has been widely considered in the literature for both monotone and
nonmonotone functions \cite{svitkina2010facility,ene2014hardness,chekuri2011submodular},
and is referred to as the {\sc Minimum Submodular Cost Allocation (MSCA) problem}\footnote{Sometimes referred to as submodular procurement auctions.} (introduced in \cite{hayrapetyan2005network,svitkina2010facility}  and further developed in \cite{chekuri2011submodular}).
\fi
\iffalse
Then first notice that by applying the results from Theorems \ref{thm:min-k-gap} and \ref{thm:curvature}
to the special setting where $\F=\{V\}$ we immediately obtain the following.

\begin{corollary}
	There is a polytime $\min \{k,n\}$-approximation algorithm for monotone MVSO($\F$) minimization over $\F=\{V\}$.
\end{corollary}
\fi
We show that the curvature dependent approximation factors obtained in Theorem \ref{thm:curvature} are essentially tight.

%Our arguments are based on the techniques used in \cite{goemans2009approximating,feige2011maximizing,svitkina2011submodular,iyer2013curvature}.
We follow a technique from \cite{goemans2009approximating,feige2011maximizing,svitkina2011submodular} and build two multivariate submodular functions that are hard to distinguish with high probability for any (even randomized) algorithm. 
%This is done as follows.

Assume that $k=n$, and let $\mathcal{R}:= (R_1,R_2,\ldots,R_n) \subseteq V^n$ be a random partition of $V$. Notice that $\sum_{i=1}^n |R_i| = n$. Let $\beta = \omega(\log n)$ and such that $\beta$ is integer. Consider the two nonnegative monotone $n$-multi-submodular functions $f_1,f_2:2^{nV} \to \R_+$ given by:
\begin{equation}
\label{eq:functions-hardness}
f_1 (S_1,\ldots,S_n) = \min \{n, \sum_{i=1}^n |S_i|\} \hspace*{0.2cm} \mbox{,} \hspace*{0.2cm}  f_2 (S_1,\ldots,S_n) = \min \{f_1 (S_1,\ldots,S_n), \beta + \sum_{i=1}^n |S_i \cap \bar{R_i}|\},
\end{equation}
where $\bar{R_i}$ denotes the complement of the set $R_i$, i.e. $\bar{R_i} = V - R_i$.

\iffalse
We start by building two functions multivariate
First, by following closely the arguments used in \cite{goemans2009approximating,feige2011maximizing,svitkina2011submodular},
%to show a bicriteria lower bound hardness for minimizing monotone submodular functions subject to a cardinality constraint, 
we prove the following hardness result for $k$-multi-submodular objectives.
 %for instances of (\ref{lifted}) where $k=n$ (i.e. where $|E|=n^2$).
\fi

The work of Svitkina and Fleischer \cite{svitkina2011submodular} show the following result for univariate functions.

\begin{lemma}[\cite{svitkina2011submodular}]
	\label{lem:S-F}
	Let $f_1$ and $f_2$ be two set functions, with $f_2$, but not $f_1$, parametrized by a string of random bits $r$. If for any set $S$, chosen without knowledge of $r$, the probability (over $r$) that $f_1(S) \neq f_2 (S)$ is $n^{-\omega(1)}$, then any algorithm that makes a polynomial number of oracle queries has probability at most $n^{-\omega(1)}$ of distinguishing $f_1$ and $f_2$.
\end{lemma}

The above clearly generalizes to the setting of tuples (i.e. multivariate objectives) in a natural and straightforward way. The only difference is that our ground set in the lifted space has now size $n^2$ instead of $n$. %However, this makes no difference, since a polynomial in $n^2$ is still a polynomial in $n$.

\begin{lemma}
	\label{lem:S-F-MV}
	Let $f_1$ and $f_2$ be two $n$-multivariate set functions, with $f_2$, but not $f_1$, parametrized by a string of random bits $r$. If for any tuple $(S_1,\ldots,S_n)$, chosen without knowledge of $r$, the probability (over $r$) that $f_1 (S_1,\ldots,S_n) \neq f_2 (S_1,\ldots,S_n)$ is $n^{-\omega(1)}$, then any algorithm that makes a polynomial number of oracle queries has probability at most $n^{-\omega(1)}$ of distinguishing $f_1$ and $f_2$.
\end{lemma}

We can use Lemma \ref{lem:S-F-MV} to show the following result for the functions defined in (\ref{eq:functions-hardness}).

\begin{lemma}
	\label{lem:MV-monot-hardness}
	Any algorithm that makes a polynomial number of oracle calls has probability $n^{-\omega(1)}$ of distinguishing the functions $f_1$ and $f_2$ above.
	\iffalse
	[NOTE: This lemma can also be written as follows: Fix an arbitrary subset $S \subseteq E$, and then let  $R$ be a random $(\infty,1)$-matching over $E$. Then the probability (over the choice of $\mathcal{R}$) that $g_1(S) \neq g_2 (S)$ is at most $N^{-\omega(1)}$.]
	\fi
\end{lemma}
\begin{proof}
	By Lemma \ref{lem:S-F-MV}  it suffices to show that for any tuple $(S_1,\ldots,S_n)$ the probability (over the random choice of the partition $\mathcal{R}$) that $f_1 (S_1,\ldots,S_n) \neq f_2 (S_1,\ldots,S_n)$ is at most $n^{-\omega(1)}$. 
	
	Let us denote this probability by $p (S_1,\ldots,S_n)$. We first show that $p (S_1,\ldots,S_n)$ is maximized for tuples $(S_1,\ldots,S_n)$ satisfying $\sum_{i=1}^n |S_i| = n$. First suppose that $\sum_{i=1}^n |S_i| > n$. Then $p (S_1,\ldots,S_n) = \mathbb{P}[\beta + \sum_{i=1}^n |S_i \cap \bar{R_i}| < n]$. But this probability can only increase if an element is removed from some set $S_i$. Similarly, in the case where $\sum_{i=1}^n |S_i| < n$, we get $p (S_1,\ldots,S_n) =\mathbb{P} [\beta + \sum_{i=1}^n |S_i \cap \bar{R_i}| < \sum_{i=1}^n |S_i|] = \mathbb{P} [\sum_{i=1}^n |S_i \cap R_i| > \beta]$. But this probability can only increase if an element is added to some set $S_i$.
		
	So let $(S_1,\ldots,S_n)$ be any fixed tuple satisfying $\sum_{i=1}^n |S_i| = n$, and let $m_v := \sum_{i:S_i\ni v} 1$ denote the number of sets $S_i$ that contain a copy of $v$. Note that $\sum_{v\in V} m_v = \sum_{i=1}^n |S_i| = n$. 
	%We show that $p (S_1,\ldots,S_n)=n^{-\omega(1)}$.  
	Let us consider a random partition $\mathcal{R}=(R_1,R_2,\ldots,R_n)$ which is obtained by placing each element $v\in V$ independently and uniformly at random into one of the sets $R_1,R_2,\ldots,R_n$. Let $X_v$ be a random variable for each $v \in V$, defined by $X_v = \sum_{i=1}^n |S_i \cap R_i \cap \{v\}|$. That is, $X_v = 1$ if $v$ is assigned to an $R_i$ such that $S_i \ni v$ (which happens with probability $m_v/n$), and $X_v = 0$ otherwise. Clearly, the random variables $\{X_v\}_{v\in V}$ are pairwise independent. Moreover, we have that the expected value of $\sum_{i=1}^n |S_i \cap R_i|$ is given by
	\begin{equation*}
	\mu :=  \mathbb{E}[\sum_{i=1}^n |S_i \cap R_i|] = \mathbb{E}[ \sum_{v\in V}X_v] = \sum_{v\in V} \mathbb{E}[X_v] = \sum_{v\in V} \frac{m_v}{n} = 1.
	\end{equation*}
	Then, by Chernoff bounds and using that $\beta$ is an integer we obtain
	\begin{align*}
	p (S_1,\ldots,S_n) = &  %\mathbb{P} [\beta + \sum_{i=1}^n |S_i \cap \bar{R_i}| < \sum_{i=1}^n |S_i|] = 
	\mathbb{P} [\sum_{i=1}^n |S_i \cap R_i| > \beta] = \mathbb{P} [\sum_{i=1}^n |S_i \cap R_i| \geq \beta + 1] \\
	= & \mathbb{P} [\sum_{v\in V} X_v \geq (1+\beta) \mu] \leq e^{- \mu \beta / 3} = e^{-\beta/3} = e^{-\omega(\log n)} = n^{-\omega(1)}.
	\end{align*} 
\end{proof}

We now prove our (curvature independent) lower bound result.

\begin{theorem}
	\label{thm:MV-monot-hardness-proof}
	The monotone MVSO($\F$) minimization problem over $\F=\{V\}$ cannot be approximated to a ratio $o(n/\log n)$ in the value oracle model with polynomially many queries.
\end{theorem} 
\begin{proof}%[Theorem \ref{thm:MV-monot-hardness-proof}]
	Assume there is a polytime algorithm achieving an approximation factor of $\alpha = o(n / \log n)$. Choose $\beta = \omega(\log n)$ such that $\alpha \beta <n$. Consider the output of the algorithm when $f_2$ is given as input. The optimal solution in this case is the partition $\mathcal{R}=(R_1,\ldots,R_k)$, with $f_2(R_1,\ldots,R_k)=\beta$. So the algorithm produces a feasible solution (i.e. a partition) $(S^*_1,\ldots,S^*_k)$ satisfying $f_2(S^*_1,\ldots,S^*_k) \leq \alpha \beta < n$. However, since $f_1$ takes value exactly $n$ over any partition, there is no feasible solution $(S_1,\ldots,S_n)$ such that $f_1 (S_1,\ldots,S_n) < n$. This means that if the input is the function $f_1$ then the algorithm produces a different answer, thus distinguishing between $f_1$ and $f_2$, contradicting Lemma \ref{lem:MV-monot-hardness}.
\end{proof}

%\noindent
The above result contrasts with the known $O(\log n)$ approximation (\cite{svitkina2010facility}) for the case where the multivariate objective is separable, that is $f \S = \sum_i f_i(S_i)$. These two facts combined now prove Theorem \ref{thm:MV-monot-hardness}.

We use a construction from Iyer et al \cite{iyer2013curvature} to explicitly introduce the effect of curvature
into the lower bound. Their work is for univariate functions, but it can be naturally extended to the
multivariate setting. We modify the functions $f_1,f_2$ from (\ref{eq:functions-hardness}) as follows:
\iffalse
The work of Iyer et al \cite{iyer2013curvature} uses a construction that enables to explicitly introduce the effect of curvature
into information-theoretic lower bounds for monotone submodular functions in the single-setting. We use a natural extension of their
construction to the multivariate setting, and consider the following modifications of the functions $f_1,f_2$
previously defined in (\ref{eq:functions-hardness}):
\fi
\begin{equation*}
\label{eq:functions-hardness-curvature}
f^c_i \S = c \cdot f_i \S + (1-c) \sum_{i=1}^n |S_i| \hspace*{0.2cm} \mbox{, for } i=1,2.
%f^c_1 \S = c f_1 \S + (1-c) \sum_{i=1}^n |S_i| \hspace*{0.2cm} \mbox{,} \hspace*{0.2cm} f^c_2 \S = c f_2 \S + (1-c) \sum_{i=1}^n |S_i|
\end{equation*}
It is then straightforward to check that both $f_1^c$ and $f_2^c$ have total curvature $c$.
Moreover, since $f_1 \S = f_2 \S$ if and only if $f^c_1 \S = f^c_2 \S$, by Lemma \ref{lem:MV-monot-hardness} it follows that any algorithm that makes polynomially many queries is not able to distinguish between $f^c_1$ and $f^c_2$ with high probability. In addition, the gap between the optimal solutions for these two functions is given by
\begin{equation*}
	\frac{OPT_1}{OPT_2}= \frac{cn + (1-c)n}{c \beta + (1-c)n} = \frac{n}{c \beta + (1-c)n} = \frac{n}{\beta + (n-\beta)(1-c)} = \frac{n/\beta}{1 + (n/ \beta-1)(1-c)}.
\end{equation*}
Then, since $\beta = \omega(\log n)$, the (curvature dependent) lower bound follows.
%Recall that this shows that the curvature dependent approximation factors obtained in Theorem \ref{thm:curvature} are essentially tight. 

\begin{reptheorem}{thm:MV-monot-hardness-curvature}
	%\label{thm:MV-monot-hardness-curvature}
	The monotone MVSO($\F$) minimization problem over $\F=\{V\}$ and objectives $f$ with total curvature $c$
	cannot be approximated to a ratio $o(\frac{n/\log n}{1+(\frac{n}{\log n}-1)(1-c)})$ in the value oracle model with polynomial number of queries.
\end{reptheorem}

\section{Conclusions}
\label{sec:conclusions}

We introduce a new class of multivariate submodular optimization problems, and give information theoretic evidence that this class encodes much
more than the separable versions arising in multi-agent objectives. We provide some explicit examples and potential applications.

For maximization, we show that practical algorithms such as accelerated greedy variants and distributed algorithms achieve good approximation guarantees under very general constraints. For arbitrary families, we show MV gaps of $1-1/e$ and $0.385$ for the monotone and nonmonotone problems respectively, and the MV gap for monotone objectives is tight.

For minimization the news is worse. However, we give (essentially tight) approximation factors with respect to the curvature of the multivariate objective function. This may lead to significant gains in several settings.

\bibliography{REFERENCES}

%\newpage
\appendix
%\begin{center}
%	\large \textbf{APPENDIX (Paper full version)}
%\end{center}

\section{Properties of $k$-multi-submodular functions}
\label{sec:properties-mv-functions}

In this section we discuss several properties of $k$-multi-submodular functions. We see that some of the characterizations and results that hold for univariate submodular functions extend naturally to the multivariate setting.

We start by showing that our definition of submodularity in the multivariate setting captures the diminishing return property.
Recall that we usually think of the pair $(i,v) \in [k] \times V$ as the assignment of element $v$ to agent $i$. We use this to introduce some notation for adding an element to a tuple.

\begin{definition}
	\label{def:mv-sum}
	Given a tuple $(S_1,\ldots,S_k) \in 2^{kV}$ and $(i,v) \in [k]\times V$, we denote by $(S_1,\ldots,S_k)+(i,v)$ the new tuple $(S_1,\ldots,S_{i-1},S_i + v, S_{i+1},\ldots,S_k)$.
\end{definition}
Then, it is natural to think of the quantity
\begin{equation}
\label{eq:mv-diminishing-returns-appendix}
f(S_1,\ldots,S_{i-1},S_i + v, S_{i+1},\ldots,S_k) - f(S_1,\ldots,S_{i-1},S_i, S_{i+1},\ldots,S_k)
\end{equation}
as the marginal gain of assigning element $v$ to agent $i$ in the tuple $(S_1,\ldots,S_k)$. Notice that with the notation introduced in Definition \ref{def:mv-sum} we have that (\ref{eq:mv-diminishing-returns-appendix}) can be also written as
$$
f((S_1,\ldots,S_k) + (i,v))- f(S_1,\ldots,S_k).
$$
This leads to the following diminishing returns characterizations in the multivariate setting.

\begin{proposition}
	\label{prop:equiv-mv-def-1}
	A multivariate function $f:2^{kV} \to \R$ is $k$-multi-submodular if and only if for all tuples $(S_1,\ldots,S_k) \subseteq (T_1,\ldots,T_k)$ and $(i,v) \in [k]\times V$ such that $v \notin T_i$ we have
	\begin{equation}
	\label{def:mv1}
	f((S_1,\ldots,S_k) + (i,v))- f(S_1,\ldots,S_k) \geq f((T_1,\ldots,T_k)+(i,v)) - f(T_1,\ldots,T_k).
	\end{equation}
\end{proposition}
\begin{proof}
	We make use of the lifting reduction presented in Section \ref{sec:lifting-reduction}.
	Let $\bar{f}:2^E \to \R$ denote the lifted function, and let $S,T \subseteq E$ be the sets in the lifted
	space corresponding to the tuples $\S$ and $\T$ respectively.
	%Let $\pi: 2^{kV} \to 2^E$ denote the bijection discussed in Section \ref{sec:lifting-reduction}, and let $\bar{f}:2^E \to \R$ be the univariate function defined by $\bar{f}(S)=f(\pi^{-1}(S))$. Also, let $S:=\pi(S_1,\ldots,S_k)$ and $T:=\pi(T_1,\ldots,T_k)$.
	Then, since $(S_1,\ldots,S_k) \subseteq (T_1,\ldots,T_k)$, we know that $S \subseteq T$. Moreover, notice that
	\begin{equation*}
	f((S_1,\ldots,S_k) + (i,v))- f(S_1,\ldots,S_k) = \bar{f}(S+(i,v))-\bar{f}(S)
	\end{equation*}
	and
	\begin{equation*}
	f((T_1,\ldots,T_k) + (i,v))- f(T_1,\ldots,T_k) = \bar{f}(T+(i,v))-\bar{f}(T).
	\end{equation*}
	In addition, from Claim \ref{claim:subm-invariance} in Section \ref{sec:lifting-reduction} we know that $f$ is $k$-multi-submodular if and only if $\bar{f}$ is submodular. Then the result follows by observing the following.
	\begin{align*}
	& f \mbox{ is $k$-multi-submodular}\\
	\iff & \bar{f} \mbox{ is submodular} \\
	\iff & \bar{f}(S+(i,v))-\bar{f}(S) \geq \bar{f}(T+(i,v))-\bar{f}(T) \quad \mbox{ for all } S \subseteq T \mbox{ and } (i,v) \notin T \\
	\iff & f((S_1,\ldots,S_k) + (i,v))- f(S_1,\ldots,S_k) \geq f((T_1,\ldots,T_k)+(i,v)) - f(T_1,\ldots,T_k)\\
	& \mbox{ for all } (S_1,\ldots,S_k) \subseteq (T_1,\ldots,T_k) \mbox{ and } v \notin T_i.
	\end{align*}
\end{proof}

The proof of the above result also shows the following characterization of $k$-multi-submodular functions.
\begin{proposition}
	\label{prop:equiv-mv-def-2}
	A multivariate function $f:2^{kV} \to \R$ is $k$-multi-submodular if and only if for all tuples $(S_1,\ldots,S_k)$ and $(i,v),(j,u) \in [k]\times V$ such that $v \notin S_i$ and $u \notin S_j$ we have
	\begin{equation}
	\label{def:mv2}
	f((S_1,\ldots,S_k) + (i,v))- f(S_1,\ldots,S_k) \geq f((S_1,\ldots,S_k)+(j,u)+(i,v)) - f((S_1,\ldots,S_k)+(j,u)).
	\end{equation}
\end{proposition}

\section{Examples of $k$-multi-submodular functions}
\label{sec:examples-mv-functions}

We now provide some explicit examples of $k$-multi-submodular functions that lead to interesting applications.

\begin{lemma}
	\label{lem:MV-multilinear}
	Consider a multilinear function $h:\Z^k_+ \to \R$ given by $h(z)=\sum_{S \subseteq [k]} a_S \prod_{m \in S} z_m$. Let $f:2^{kV} \to \R$ be a multivariate set function defined as $f(S_1,\ldots,S_k)=h(|S_1|,\ldots,|S_k|)$. Then, $f$ is $k$-multi-submodular if and only if
	\begin{equation}
	\label{eq:MV-multilinear-conditions}
	a_S \leq 0 \quad \forall S \subseteq [k].
	\end{equation}
\end{lemma}
\begin{proof}
	By Proposition \ref{prop:equiv-mv-def-2} we know that $f$ is $k$-multi-submodular if and only if
	%for all tuples $\S$ and $(i,v),(j,u) \in [k]\times V$
	condition (\ref{def:mv2}) is satisfied. Let $\S$ be an arbitrary tuple and let $(i,v),(j,u) \in [k]\times V$ such that $v \notin S_i, u \notin S_j$. Denote by $z^0$ the integer vector with components $z^0_i = |S_i|$. That is, $z^0 = (|S_1|,|S_2|,\ldots,|S_k|) \in \Z^k_+$. We call $z^0$ the cardinality vector associated to the tuple $\S$. In a similar way, let $z^1$ be the cardinality vector associated to the tuple $\S + (i,v)$, $z^2$ the cardinality vector associated to $\S + (j,u)$, and $z^3$ the cardinality vector associated to $\S + (i,v) + (j,u)$. Now notice that condition (\ref{def:mv2}) can be written as
	\begin{equation}
	\label{eq:MV-multilinear-conditions1}
	h(z^1) - h(z^0) \geq h(z^3) - h(z^2)
	\end{equation}
	for all $z^0,z^1,z^2,z^3 \in \Z^k_+$ such that $z^1 = z^0 + \bf{e_i}$, $z^2 = z^0 + \bf{e_j}$, and $z^3 = z^0 + \bf{e_i} + \bf{e_j}$, where $\bf{e_i}$ is the characteristic vector on the $ith$ component, and similarly for $\bf{e_j}$.

	We show that (\ref{eq:MV-multilinear-conditions1}) is equivalent to (\ref{eq:MV-multilinear-conditions}).
	Using that $z^1_m = z^0_m$ for all $m \neq i$ and $z^1_i = z^0_i + 1$, we have
	\begin{align*}
	h(z^1) - h(z^0) =& \sum_{S \subseteq [k]} a_S \prod_{m \in S} z^1_m - \sum_{S \subseteq [k]} a_S \prod_{m \in S} z^0_m\\
	=& \sum_{S \subseteq [k]} a_S [\prod_{m \in S} z^1_m - \prod_{m \in S} z^0_m]\\
	=& \sum_{S \ni i} a_S [\prod_{m \in S} z^1_m - \prod_{m \in S} z^0_m]\\
	=& \sum_{S \ni i} a_S [(z^0_i+1)\prod_{m \in S, m \neq i} z^0_m - \prod_{m \in S} z^0_m]\\
	=& \sum_{S \ni i} a_S \prod_{m \in S, m \neq i} z^0_m.
	\end{align*}
	Similarly, using that $z^3_m = z^2_m$ for all $m \neq i$ and $z^3_i = z^2_i + 1$, we have
	\begin{align*}
	h(z^3) - h(z^2) =& \sum_{S \ni i} a_S \prod_{m \in S, m \neq i} z^2_m\\
	=& \sum_{S \ni i, S \ni j} a_S \prod_{m \in S, m \neq i} z^2_m + \sum_{S \ni i, S \notni j} a_S \prod_{m \in S, m \neq i} z^2_m\\
	=& \sum_{S \ni i, S \ni j} a_S (z^0_j +1) \prod_{m \in S, m \neq i,j} z^0_m + \sum_{S \ni i, S \notni j} a_S \prod_{m \in S, m \neq i} z^0_m\\
	=&\sum_{S \ni i} a_S \prod_{m \in S, m \neq i} z^0_m + \sum_{S \ni i, S \ni j} a_S \prod_{m \in S, m \neq i,j} z^0_m\\
	=&h(z^1)-h(z^0)+\sum_{S \ni i, S \ni j} a_S \prod_{m \in S, m \neq i,j} z^0_m,
	\end{align*}
	where in the third equality we use that $z^2 = z^0 + \bf{e_j}$. Thus, we have
	\begin{equation*}
	h(z^1) - h(z^0) \geq h(z^3) - h(z^2)
	\iff \sum_{S \ni i, S \ni j} a_S \prod_{m \in S, m \neq i,j} z^0_m \leq 0.
	\end{equation*}
	Since the above must hold for all $z^0,z^1,z^2,z^3 \in \Z^k_+$ and $i,j \in [k]$ such that $z^1 = z^0 + \bf{e_i}$, $z^2 = z^0 + \bf{e_j}$, and $z^3 = z^0 + \bf{e_i} + \bf{e_j}$, we immediately get that (\ref{eq:MV-multilinear-conditions1}) is equivalent to (\ref{eq:MV-multilinear-conditions}) as we wanted to show.
\end{proof}

\begin{lemma}
	\label{lem:MIMO}
	Consider a quadratic function $h:\Z^k_+ \to \R$ given by $h(z)=z^T A z$ for some matrix $A=(a_{ij})$. Let $f:2^{kV} \to \R$ be a multivariate set function defined as $f(S_1,\ldots,S_k)=h(|S_1|,\ldots,|S_k|)$. Then, $f$ is $k$-multi-submodular if and only if $A$ satisfies
	\begin{equation}
	\label{eq:MIMO-conditions}
	a_{ij} + a_{ji} \leq 0 \quad \forall i,j \in [k].
	\end{equation}
\end{lemma}
\begin{proof}
	The proof is very similar to that of Lemma  \ref{lem:MV-multilinear}.
	By Proposition \ref{prop:equiv-mv-def-2} we know that $f$ is $k$-multi-submodular if and only if
	%for all tuples $\S$ and $(i,v),(j,u) \in [k]\times V$
	condition (\ref{def:mv2}) is satisfied. Let $\S$ be an arbitrary tuple and let $(i,v),(j,u) \in [k]\times V$ such that $v \notin S_i, u \notin S_j$. Denote by $z^0$ the integer vector with components $z^0_i = |S_i|$. That is, $z^0 = (|S_1|,|S_2|,\ldots,|S_k|) \in \Z^k_+$. We call $z^0$ the cardinality vector associated to the tuple $\S$. In a similar way, let $z^1$ be the cardinality vector associated to the tuple $\S + (i,v)$, $z^2$ the cardinality vector associated to $\S + (j,u)$, and $z^3$ the cardinality vector associated to $\S + (i,v) + (j,u)$. Now notice that condition (\ref{def:mv2}) can be written as
	\begin{equation}
	\label{eq:MIMO1}
	h(z^1) - h(z^0) \geq h(z^3) - h(z^2)
	\end{equation}
	for all $z^0,z^1,z^2,z^3 \in \Z^k_+$ such that $z^1 = z^0 + \bf{e_i}$, $z^2 = z^0 + \bf{e_j}$, and $z^3 = z^0 + \bf{e_i} + \bf{e_j}$, where $\bf{e_i}$ is the characteristic vector on the $ith$ component, and similarly for $\bf{e_j}$.

	We show that (\ref{eq:MIMO1}) is equivalent to (\ref{eq:MIMO-conditions}).
	First notice that for a vector $z=(z_1,\ldots,z_k)$ the function $h$ can be written as $h(z)=\sum_{\ell,m=1}^k a_{\ell m} z_\ell z_m$. Then, using that $z^1_\ell = z^0_\ell$ for all $\ell \neq i$ and $z^1_i = z^0_i + 1$, we have
	\begin{equation*}
	h(z^1) - h(z^0) = \sum_{\ell,m=1}^k a_{\ell m} z^1_\ell z^1_m - \sum_{\ell,m=1}^k a_{\ell m} z^0_\ell z^0_m = \sum_{\ell=1}^k a_{\ell i} z^0_\ell + \sum_{m=1}^k a_{im} z^0_m + a_{ii}.
	\end{equation*}
	\iffalse
	By the symmetry of $A$ the above expression can be further simplified to
	\begin{equation*}
	h(z^1) - h(z^0) = a_{ii} + 2 \sum_{l=1}^k a_{li} z^0_l.
	\end{equation*}
	\fi
	Similarly, using that $z^3_\ell = z^2_\ell$ for all $\ell \neq i$ and $z^3_i = z^2_i + 1$, we have
	\begin{equation*}
	h(z^3) - h(z^2) = \sum_{\ell=1}^k a_{\ell i} z^2_\ell + \sum_{m=1}^k a_{im} z^2_m + a_{ii}.
	\end{equation*}
	\iffalse
	There are two cases to consider, either $i=j$ or $i\neq j$. In the latter we have that $z^2_i = z^0_i$, and hence $h(z^1) - h(z^0) = h(z^3) - h(z^2)$. So in the case $i \neq j$ condition (\ref{eq:MIMO1}) is always satisfied. If $i=j$, then $z^2_i = z^0_i + 1$.
	\fi
	Thus, using that $z^2 = z^0 + \bf{e_j}$ we get
	\begin{align*}
	& h(z^1) - h(z^0) \geq h(z^3) - h(z^2) \\
	\iff & \sum_{\ell=1}^k a_{\ell i} z^0_\ell + \sum_{m=1}^k a_{im} z^0_m + a_{ii} \;\geq \; \sum_{\ell=1}^k a_{\ell i} z^2_\ell + \sum_{m=1}^k a_{im} z^2_m + a_{ii} \\
	\iff & \sum_{\ell=1}^k a_{\ell i}(z^0_\ell - z^2_\ell) + \sum_{m=1}^k a_{im}(z^0_m - z^2_m)  \geq 0 \\
	\iff & - a_{ji} - a_{ij} \geq 0 \\
	\iff & a_{ji} + a_{ij} \leq 0.
	\end{align*}
\end{proof}

\section{Upwards-closed (aka blocking) families}
\label{sec:blocking}

In this section, we give some background for blocking families.
As our work for minimization is restricted to  monotone functions, we can often convert an arbitrary
set family into its upwards-closure (i.e., a blocking version of it) and work with it instead. We discuss this
reduction as well.
The technical details discussed in this section are fairly standard and we include them for completeness.
%Several of these results have already appeared in \cite{iyer2014monotone}.

%That is, families which are equal to their {\em upwards closure}:  $\F^\uparrow  \{F:  \exists  ~F' \in \F ~s.t. F' \subseteq F\}$.

A set family $\F$ over a ground set $V$ is {\em upwards-closed}
if $F \subseteq F'$ and $F \in \F$, implies that $F' \in \F$; these are sometimes referred to as  {\em blocking families}.
Examples of such families include vertex covers or set covers more generally,  whereas spanning trees are not.

\subsection{Reducing to blocking families}

Now consider an arbitrary set family $\F$ over $V$. We may define its {\em upwards closure}
by $\F^{\uparrow}=\{F':  F \subseteq F' \textit{ for some $F \in \F$}\}$. In this section we argue that
in order to solve a monotone optimization problem over sets in $\F$ it is often sufficient to work
over its upwards-closure.

%As already  noted $\B(\F)=\B(\F^{\uparrow}) = \B (\F^{min})$ and hence one approach
%is via the blocking formulation $P^*(\F)=P^*(\F^{\uparrow})$.

This requires two ingredients. First, we  need
a separation algorithm for the relaxation $P^*(\F)$, but indeed this is often available for many
natural families such as spanning trees, perfect matchings, $st$-paths, and vertex covers.   The second ingredient  needed is the ability to turn an integral solution $\chi^{F'}$ from  $P^*(\F^{\uparrow})$ or $P(\F^\uparrow)$
into an integral solution $\chi^F \in P(\F)$. We now argue that this is the case if a polytime separation algorithm is available for $P^*(\F^{\uparrow})$ or for the polytope $P(\F):= conv(\{\chi^F: \textit{$F \in \F$}\})$.

\iffalse
doable if we have separation over the polyhedron $P(\F)=conv(\chi^F: F \in \F)$.
\fi

For a polyhedron $P$,  we denote its {\em dominant} by
$P^{\uparrow} := \{z:  z \geq x \textit{~for some~} x \in P \}$.
The following observation is straightforward.

\begin{claim}
	\label{claim:lattice-points}
	Let $H$ be the set of vertices of the hypercube in $\R^V$. Then
	$$H \cap P(\F^\uparrow) = H \cap P(\F)^\uparrow = H \cap P^*(\F^\uparrow).$$
	In particular we have that $\chi^S \in P(\F)^\uparrow \iff \chi^S \in P^*(\F^\uparrow)$.
\end{claim}

\iffalse
\begin{claim}
	If we can separate over $P$ in polytime,  then we can  separate over the dominant $P^\uparrow$.
\end{claim}
\begin{proof}
	Given a vector $y$, we can decide whether $y \in P^\uparrow$ by solving
	\begin{align*}
	x + s = y \\
	x \in P \\
	s \geq 0.
	\end{align*}
	Since can we easily separate over the first and third constraints, and a separation oracle for $P$ is given (i.e. we can also separate over the set of constraints imposed by the second line), it follows that we can separate over the above set of constraints in polytime.
\end{proof}

\begin{corollary}
	If  we can separate over $P(\F)$,
	then we can separate over the dominant $P(\F)^\uparrow \subseteq P^*(\F)$.
\end{corollary}
\fi

We can now use this observation to prove the following.

\iffalse
This yields the following mechanism for turning feasible sets in $\F^{\uparrow}$ into feasible sets in $\F$.
\fi

\begin{lemma}
	\label{lem:dominant-reduction}
	Assume we have a separation algorithm for $P^*(\F^\uparrow)$. Then for any $\chi^{S} \in P^*(\F^\uparrow)$ we can find in polytime $\chi^{M} \in P(\F)$ such that $\chi^{M} \leq \chi^{S}$.
\end{lemma}
\begin{proof}
	Let $S=\{1,2,\ldots,k\}$. We run the following routine until no more elements can be removed:
	\vspace*{10pt}
	
	For $i \in S$\\
	\hspace*{20pt} If $\chi^{S-i} \in P^*(\F^\uparrow)$ then $S=S-i$
	
	\vspace*{10pt}
	Let $\chi^M$ be the output. We show that $\chi^M \in P(\F)$. Since $\chi^M \in P^*(\F^\uparrow)$, by Claim \ref{claim:lattice-points} we know that $\chi^M \in P(\F)^\uparrow$. Then by definition of dominant there exists $x\in P(\F)$ such that $x\leq \chi^M \in P(\F)^\uparrow$. It follows that the vector $x$ can be written as $x = \sum_{i} \lambda_{i} \chi^{U_i}$ for some $U_i \in \F$ and $\lambda_i \in (0,1]$ with $\sum_i \lambda_i =1$. Clearly we must have that $U_i \subseteq M$ for all $i$, otherwise $x$ would have a non-zero component outside $M$.  In addition, if for some $i$ we have $U_i \subsetneq M$, then there must exist some $j \in M$ such that $U_i \subseteq M-j \subsetneq M$. Hence $M-j \in \F^{\uparrow}$, and thus $\chi^{M-j} \in P(\F)^\uparrow$ and $\chi^{M-j} \in P^*(\F^\uparrow)$. But then when component $j$ was considered in the algorithm above, we would have had $S$ such that $M \subseteq S$ and so $\chi^{S-j} \in P^*(\F^\uparrow)$ (that is $\chi^{S-j} \in P(\F)^\uparrow$), and so $j$ should have been removed from $S$, contradiction.
\end{proof}

We point out that for many natural set families $\F$ we can work with the relaxation $P^*(\F^\uparrow)$ assuming that it admits a separation algorithm.
Then, if we have an algorithm which produces  $\chi^{F'} \in P^*(\F^\uparrow)$ satisfying some approximation guarantee for a monotone problem, we can use Lemma \ref{lem:dominant-reduction} to construct in polytime $F \in \F$ which obeys the same guarantee.

Moreover, notice that for Lemma \ref{lem:dominant-reduction} to work we do not need an actual separation oracle for $P^*(\F^\uparrow)$, but rather all we need is to be able to separate over $0-1$ vectors only. Hence, since the polyhedra $P^*(\F^\uparrow), \, P(\F^\uparrow)$ and $P(\F)^\uparrow$ have the same $0-1$ vectors (see Claim \ref{claim:lattice-points}), a separation oracle for either $P(\F^\uparrow)$ or $P(\F)^\uparrow$ would be enough for the routine of Lemma \ref{lem:dominant-reduction} to work. We now show that this is the case if we have a polytime separation oracle for $P(\F)$. The following result shows that if we can separate efficiently over $P(\F)$ then we can also separate efficiently over the dominant $P(\F)^\uparrow$.

\iffalse
We now show that we can also turn in polytime an integral solution $\chi^{F'} \in P(\F^\uparrow)$ into an integral solution $\chi^F \in P(\F)$. We work from the point of view of the dominant $P(\F)^\uparrow$. Notice that this is fine since by Claim \ref{claim:lattice-points} we have that $\chi^{F'} \in P(\F^\uparrow) \iff \chi^{F'} \in P(\F)^\uparrow$. The following result shows that if we can separate efficiently over $P(\F)$ then we can also separate efficiently over the dominant $P(\F)^\uparrow$.
\fi

\begin{claim}
	If we can separate over a polyhedron $P$ in polytime,  then we can also separate over its dominant $P^\uparrow$ in polytime.
\end{claim}
\begin{proof}
	Given a vector $y$, we can decide whether $y \in P^\uparrow$ by solving
	\begin{align*}
	x + s = y \\
	x \in P \\
	s \geq 0.
	\end{align*}
	Since can we easily separate over the first and third constraints, and a separation oracle for $P$ is given (i.e. we can also separate over the set of constraints imposed by the second line), it follows that we can separate over the above set of constraints in polytime.
\end{proof}

Now we can apply the same mechanism from Lemma \ref{lem:dominant-reduction} to turn feasible sets from $\F^{\uparrow}$ into feasible sets in $\F$.

\begin{corollary}
	\label{cor:dominant-reduction2}
	Assume we have a separation algorithm for $P(\F)^\uparrow$. Then for any $\chi^{S} \in P(\F)^\uparrow$ we can find in polytime $\chi^{M} \in P(\F)$ such that $\chi^{M} \leq \chi^{S}$.
\end{corollary}

We conclude this section by making the remark that if we have an algorithm which produces  $\chi^{F'} \in P(\F^\uparrow)$ satisfying some approximation guarantee for a monotone problem, we can use Corollary \ref{cor:dominant-reduction2} to construct $F \in \F$ which obeys the same guarantee.

\end{document}